\pdfoutput=1

\documentclass[journal]{IEEEtran}

\usepackage{amssymb}
\usepackage{amsthm}
\usepackage[usenames,dvipsnames]{color}
\usepackage{setspace}
\usepackage{amsfonts}
\usepackage{amssymb}
\usepackage{amsmath}
\usepackage{booktabs}
\usepackage{dsfont}
\usepackage{enumerate}
\usepackage{times} % or times
\usepackage{todonotes}
\usepackage{latexsym}
\usepackage{pgf,pgfarrows,pgfnodes,pgfautomata,pgfheaps}
\usepackage{tikz}
\usetikzlibrary{arrows,automata,fit,matrix}
\usepackage{color}
\usepackage{algorithm,algorithmic}
\usepackage{extarrows}

\usepackage{relsize}

\usepackage{graphicx}
\usepackage{subfig}
\usepackage{epstopdf}

\newtheorem{definition}{Definition}

\newtheorem{proposition}{Proposition}
\newtheorem{theorem}{Theorem}
\newtheorem{remark}{Remark}

\newtheorem{example}{Example}
\newcommand{\act}[1]{\xlongrightarrow{#1}}          % transition of with action type [1]

\newcommand{\norm}[1]{\lVert #1 \rVert_\infty}

\newcommand{\normo}[1]{\lVert #1 \rVert_1}
\newcommand{\tha}{\hat{t}}
\newcommand{\hh}{\hat{H}}
\newcommand{\hr}{\hat{r}}
\newcommand{\tr}{\tilde{r}}
\newcommand{\RE}{\mathbb{R}}
\newcommand{\REz}{\mathbb{R}_{\geq0}}

\newcommand{\calE}{\mathcal{E}}
\newcommand{\calU}{\mathcal{U}}
\newcommand{\calF}{\mathcal{F}}

\newcommand{\calR}{\mathcal{R}}
\newcommand{\calC}{\mathcal{C}}
\newcommand{\calK}{\mathcal{K}}
\newcommand{\calP}{\mathcal{P}}
\newcommand{\calI}{\mathcal{I}}

\newcommand{\cals}{\mathcal{S}}
\newcommand{\calW}{\mathcal{W}}

\DeclareMathOperator*{\argmax}{arg\,max}

\newcommand{\fru}{\mathfrak{u}}
\newcommand{\frut}{\mathfrak{\tilde{u}}}

\newcommand{\eps}{\varepsilon}

\newcommand{\epso}{\varepsilon^{\mathit{old}}}
\newcommand{\epsn}{\varepsilon^{\mathit{new}}}
\newcommand{\hk}{\hat{\kappa}}
\newcommand{\frb}{\mathfrak{b}}

\newcommand{\dt}{\Delta t}

\allowdisplaybreaks[1]

\begin{document}

\title{Over-Approximation of Fluid Models}

\author{Max~Tschaikowski
\thanks{The author is with the Technische Universit\"{a}t Wien, Austria (e-mail: max.tschaikowski@tuwien.ac.at).}}

\maketitle

\begin{abstract}
Fluid models are a popular formalism in the quantitative modeling of biochemical systems and analytical performance models. The main idea is to approximate a large-scale Markov chain by a compact set of ordinary differential equations (ODEs). Even though it is often crucial for a fluid model under study to satisfy some given properties, a formal verification is usually challenging. This is because parameters are often not known precisely due to finite-precision measurements and stochastic noise. In this paper, we present a novel technique that allows one to efficiently compute formal bounds on the reachable set of time-varying nonlinear ODE systems that are subject to uncertainty. To this end, we a) relate the reachable set of a nonlinear fluid model to a family of inhomogeneous continuous time Markov decision processes and b) provide optimal and suboptimal solutions for the family by relying on optimal control theory. The proposed technique is efficient and can be expected to provide tight bounds. We demonstrate its potential by comparing it with a state-of-the-art over-approximation approach.
\end{abstract}


\section{Introduction}\label{sec_intro}

In the last decades, fluid (or mean-field) models underlying biochemical and computer systems have gained a lot of momentum. Possible examples are chemical reaction networks~\cite{ccdtt-sci-reports-2012}, optical switches~\cite{DBLP:conf/sigmetrics/HoudtB12} and layered queueing networks~\cite{DBLP:journals/tse/Tribastone13}. The main idea is to approximate the original stochastic model which is usually given in terms of a large-scale continuous time Markov chain (CTMC) by means of a compact system of ordinary differential equations (ODEs). When the number of agents (molecules, jobs, nodes etc.) present in the system tends to infinity, the simulation runs of a suitably scaled version of the CTMC can be shown to converge in probability to the deterministic solution of the underlying fluid model~\cite{kurtz-chem,DBLP:journals/pe/BortolussiHLM13}. The law of mass action from chemistry~\cite{cttvPNAS}, for instance, has been shown to be the fluid model of a CTMC semantics stated on the molecule level~\cite{kurtz-chem}.

Unfortunately, a precise parameterization of a fluid model is often not possible due to finite-precision measurements or stochastic noise~\cite{dsn16BortolussiGast,DBLP:journals/tac/TschaikowskiT16}. Hence, in order to verify that a nonlinear fluid model satisfies some given property in the presence of parameter functions that are subject to uncertainty, it becomes necessary to estimate the reachable set. This is because a finite set of possible ODE solutions (i.e., a proper subset of the reachable set) can only be used establish the presence of property violations but does not suffice to exclude their existence in general~\cite{journals/automatica/Prajna06,DBLP:journals/automatica/AbatePLS08}. Another reason is that closed-form expressions for reachable sets of nonlinear ODE systems are not known in general~\cite{Lafferriere1999}.

The estimation of reachable sets of continuous ODE systems is of crucial importance in the field of control engineering and has received a lot of attention over the decades. Linear ODE systems with uncertainties (alternatively, disturbances) are well-understood because in this case the reachable set can be shown to be convex. The situation where also matrix coefficients are uncertain~\cite{Althoff2011a}, however, is more challenging than the standard control theoretical setting of additive uncertainties~\cite{Kurzhanski2000,Girard2006,B-GirLeG08a,Bak:2017:HTC:3049797.3049808}. Bounding the reachable set of nonlinear ODE systems is more difficult and there is a number of different techniques which complement each other. The abstraction approach approximates a nonlinear ODE system locally by an affine mapping~\cite{Asarin2003,Donze2007,Althoff2013a} or a multivariate polynomial~\cite{Berz1998,DBLP:conf/cav/ChenAS13}. The error can then be estimated using Taylor approximation and interval arithmetic~\cite{Althoff2013a,Scott201393}. While abstraction techniques can cover many practical models, in general it is computationally prohibitive to obtain tight over-approximations for larger nonlinear systems~\cite{Althoff2013a,DBLP:conf/cav/Duggirala016}. Lyapunov-like functions~\cite{journals/automatica/Prajna06,DBLP:journals/tac/Angeli02,6160735,B-GirPap06c,Fan2016} known from the stability theory of ODE systems provide an alternative to abstraction techniques. Despite the fact that they often lead to tight bounds, their automatic computation is only possible in special cases~\cite{B-GirPap06c}. In~\cite{bayen2002guaranteed,DBLP:journals/automatica/Lygeros04,mitchell2005time} it has been observed that over-approximation can be encoded as an optimal control problem. While theoretically appealing, the approach relies on the Hamilton-Jacobi equation, a partial differential equation which can only be solved for dynamical systems with few variables~\cite{Liberzon}. In a similar vein of research,~\cite{dsn16BortolussiGast,doi:10.1137/0325010} used the necessary optimality conditions of Pontryagin's principle~\cite{Liberzon} to derive heuristic estimations on reachable sets of nonlinear ODE systems. %Despite the fact that the SAT problem is NP-complete, the availability of efficient satisfiability modulo theories (SMT) solvers motivated to encode over-approximation of dynamical systems as SMT problems~\cite{DBLP:conf/tacas/KongGCC15}.

%Monotonic systems~\cite{Ramdani2008,RAMDANI2010263}, differential inequalities~\cite{DBLP:journals/tac/TschaikowskiT16} and interval arithmetic~\cite{Scott201393}, instead, allow for efficiently computable bounds. However, the quality of the approximation depends on the model and may be loose in some cases~\cite{dsn16BortolussiGast}.

\emph{Contributions.} In the present paper, we introduce an over-approximation technique for the reachable sets of fluid models. The main idea is to exploit the fact that nonlinear fluid models can be related to the linear Kolmogorov equations of CTMCs. More specifically, the technical novelty of the present work is to prove that $i)$ a nonlinear fluid model can be over-approximated by solving a family \emph{inhomogeneous} continuous time Markov decision processes (ICTMDPs)~\cite{hernandez2009} with continuous action spaces and; $ii)$ to show that the family of ICTMDPs can be solved efficiently by modifying the strict version of Pontryagin's principle~\cite{KAMIEN} which is sufficient for optimality. This allows one to estimate the reachable set of an, in general, \emph{nonlinear} ODE system by studying the reachable sets of a family of \emph{linear} ODE systems.

%Thanks to this, it is possible to estimate the reachable set of the original nonlinear fluid model by solving a family of \emph{inhomogeneous} continuous time Markov decision processes (ICTMDPs)~\cite{hernandez2009} with continuous action spaces. This allows one to estimate the reachable set of an, in general, \emph{nonlinear} ODE system by studying the reachable sets of a family of \emph{linear} ODE systems. To ensure efficient computation and tight bounds, we provide a solution to a class of ICTMDPs by modifying the strict version of Pontryagin's principle~\cite{KAMIEN} which is sufficient for optimality.

For nonlinear fluid models, the proposed approach a) is efficient; b) induces bounds that can be expected to be tight and; c) allows for an algorithmic treatment in the case where the ODE system is given by multivariate polynomials, thus covering in particular biochemical models. A comparison with the state-of-the-art tool for reachability analysis CORA~\cite{Althoff2015a} in the context of the well-known SIRS model from epidemiology~\cite{dsn16BortolussiGast} confirms the potential of the proposed technique.

%\emph{Comparison to existing approaches.} The proposed approach is complementary to existing over-approximation techniques. Indeed, while our approach is less efficient than those based on differential inequalities~\cite{Ramdani2008,RAMDANI2010263}, it can be expected to provide tighter bounds because it is relies on a strict version of Pontryagin's principle. Approaches based on abstraction~\cite{DBLP:conf/cav/ChenAS13,Althoff2015a}, optimal control theory~\cite{DBLP:journals/automatica/Lygeros04,mitchell2005time} and SMT solvers~\cite{DBLP:conf/tacas/KongGCC15}, instead, may suffer from the curse of dimensionality, while our technique does not. However, it is fair to say that abstraction approaches can be expected to outperform the present technique for ODE models of small and moderate size. Also, the proposed technique considers \emph{time-varying uncertain parameters} and \emph{fixed} initial conditions, while most approaches allow the initial conditions to be uncertain as well. We argue however that uncertainty in parameters is often sufficient in practice. In the case of biochemistry, for instance, initial concentrations can often be measured, while reaction rate coefficients are usually difficult to obtain and may vary with time~\cite{VANLIER2013305}.

\emph{Related work on CTMDPs.} With efficient solution techniques dating back to the sixties, CTMDPs~\cite{hernandez2009} belong to one of the best studied classes of optimization problems. While there exists a large body of literature on homogeneous CTMDPs, however, much less is known about the inhomogeneous case. Moreover, most works on CTMDPs interpret controls as policies, meaning that only a subclass of uncertainties is admissible~\cite[Section 8.3]{rieder2011}. Additionally, the cost function of interest is often either the discounted or the average cost~\cite{hernandez2009} which cannot be used for the over-approximation of Kolmogorov equations. To the best of our knowledge, the only work which has studied the case of inhomogeneous \text{CTMDPs} featuring continuous action spaces and time dependent policies with respect to a cost which can be used for over-approximation is~\cite{doi:10.1287/educ.1100.0077}. In this work, three concrete queueing systems were analyzed using Pontryagin's principle~\cite{Liberzon}. For each of the three models, the underlying necessary conditions were shown to be already sufficient for optimality.

\emph{Paper outline.} Section~\ref{sec_nut} provides a high-level discussion of our approach using a concrete example. Section~\ref{sec_an} continues by introducing agent networks, a rich class of ODE systems that can be covered by our technique. In Section~\ref{sec_ra} we first relate the reachable set of the original nonlinear ODE system underlying an agent network to the solution of a family of ICTMDPs. Afterwards, we present in Sections~\ref{sec_mdps}\---\ref{sec_impl} an efficient solution approach to \mbox{ICTMDP}s. Section~\ref{sec_case_studies} compares a prototype implementation of the approach with CORA, while Section~\ref{sec_disc} discusses how the approach complements existing approximation techniques. Section~\ref{sec_conclusion} concludes the paper. % After outlining the strengths and weaknesses of the presented approach in Section~\ref{sec_discussion}

\emph{Notation.} For nonempty sets $A$ and $\calI$, let $A^\calI$ denote the set of all functions from $\calI$ to $A$. Note that elements of $A^\calI$ can be interpreted as vectors with values in $A$ and coordinates in $\calI$. We write $x \leq x'$ for $x,x' \in \RE^\calI$ whenever $x_i \leq x'_i$ for all $i \in \calI$. The equality of two functions $f$ and $g$, instead, is denoted by $f \equiv g$. By $\cals$ we refer to the finite set of (agent) states; elements $V \in \RE^\cals$ of the reachable set of an ODE system are called concentrations instead. The derivative with respect to time of a function $V \in [0;T] \to \RE^\cals$ is denoted by $\dot{V}$. Instead, $\mathds{1}$ denotes the characteristic function.

\section{The Main Idea in a Nutshell}\label{sec_nut}

We first discuss the problem and the proposed solution on the example of the SIRS model from epidemiology~\cite{dsn13IacobelliTribastone} that is given by the nonlinear ODE system
\begin{align*}
\dot{V}^{\kappa_\beta}_S & = - V^{\kappa_\beta}_S V^{\kappa_\beta}_I + V^{\kappa_\beta}_R \\
\dot{V}^{\kappa_\beta}_I & = - \kappa_\beta V^{\kappa_\beta}_I + V^{\kappa_\beta}_S V^{\kappa_\beta}_I \nonumber \\
\dot{V}^{\kappa_\beta}_R & = - V^{\kappa_\beta}_R + \kappa_\beta V^{\kappa_\beta}_I , \nonumber
\end{align*}
where $V^{\kappa_\beta}_S, V^{\kappa_\beta}_I$ and $V^{\kappa_\beta}_R$ refers to the concentration of susceptible, infected and recovered agents, respectively, and $\kappa_\beta$ denotes the positive \emph{time-varying} recovery rate parameter. We are interested in the case where the parameter function $\kappa_\beta$ is \emph{uncertain}. More specifically, we assume that $\kappa_\beta \equiv \hk_\beta + u_\beta$, where $\hk_\beta$ is a known function resembling the \emph{nominal} (or average) recovery parameter function, while $u_\beta$ is an unknown \emph{uncertainty} which satisfies $|u_\beta(\cdot)| \leq \delta_\beta(\cdot)$ for some known function $\delta_\beta$. With this, the above ODE system rewrites to
\begin{align}\label{eq_nut1}
\dot{V}^{u_\beta}_S & = - V^{u_\beta}_S V^{u_\beta}_I + V^{u_\beta}_R \\
\dot{V}^{u_\beta}_I & = - (\hk_\beta + u_\beta) V^{u_\beta}_I + V^{u_\beta}_S V^{u_\beta}_I \nonumber \\
\dot{V}^{u_\beta}_R & = - V^{u_\beta}_R + (\hk_\beta + u_\beta) V^{u_\beta}_I \nonumber
\end{align}
The \emph{nominal} solution $V^0$ corresponds to the case where $\kappa_\beta \equiv \hk_\beta$, i.e., when $u_\beta \equiv 0$. In practice, \emph{nominal} parameter functions arise from finite-precision measurements, average behavior etc., while \emph{uncertainties} account for the precision of measurements, conservative parameter estimations and stochastic noise.

\emph{\textbf{Problem to solve.}} For a given time horizon $T > 0$, we seek to provide, for each $0 \leq t \leq T$, a superset which contains the reachable set $\calR(t) = \{V^{u_\beta}(t) \mid |u_\beta(\cdot)| \leq \delta_\beta(\cdot) \}$. To this end, we bound the maximal deviation of $V^{u_\beta}$ from the nominal trajectory $V^0$, i.e., for each $B \in \{S,I,R\}$, we formally estimate the function
\begin{align}\label{eq_eps_t_0}
\calE_B(t) = \sup \{ |V_B^{u_\beta}(t) - V_B^0(t)| \ \mid \ |u_\beta(\cdot)| \leq \delta_\beta(\cdot) \}
\end{align}
Since $V_B^{u_\beta}(t) = V_B^0(t) + V_B^{u_\beta}(t) - V_B^0(t)$, we infer
\[
V_B^0(t) - \calE_B(t) \leq V_B^{u_\beta}(t) \leq V_B^0(t) + \calE_B(t)
\]
for all $0 \leq t \leq T$ and $B \in \{S,I,R\}$. With this, it holds that
\[
\calR(t) \subseteq \prod_{B \in \{S,I,R\}} \big[ V_B^0(t) - \mathcal{E}_B(t); V_B^0(t) + \mathcal{E}_B(t) \big] ,
\]
i.e., $\calR(t)$ can be estimated by bounding the positive function $\mathcal{E} = (\mathcal{E}_S(\cdot),\mathcal{E}_I(\cdot),\mathcal{E}_R(\cdot))$. In what follows, we present a technique addressing this task.

\emph{\textbf{First Step: Decoupling.}} Since the formal estimation of nonlinear dynamical systems is difficult, we relate the solution of~(\ref{eq_nut1}) to that of a special linear ODE system. More specifically, we relate~(\ref{eq_nut1}) to the linear Kolmogorov equations of a suitable CTMC. To this end, we first note that~(\ref{eq_nut1}) is induced by the law of mass action~\cite{DBLP:conf/concur/CardelliTTV15} and the chemical reactions
\begin{align}\label{ex_sir_2}
S + I & \act{1} I + I, & I & \act{\hk_\beta + u_\beta} R, &  R & \act{1} S
\end{align}
The first reaction of~(\ref{ex_sir_2}) states that an infected agent can infect a susceptible one, while the second reaction implies that an infected agent eventually recovers. Instead, the third reaction expresses the fact that a recovered agent eventually loses its immunity and becomes susceptible again.

Apart from inducing the ODE system~(\ref{eq_nut1}), the chemical reactions~(\ref{ex_sir_2}) induce also a probabilistic model. Intuitively, given a large group of agents interacting according to~(\ref{ex_sir_2}), the stochastic behavior of a single agent in the group is given in terms of a CTMC with the states $S,I,R$ such that at time $t$ the transition rate
\begin{align}\label{eq_nut_trans}
& \quad \bullet \text{ from state $S$ into state $I$ is $V^{u_\beta}_I(t)$;} \\
& \quad \bullet \text{ from state $I$ into state $R$ is $\hk_\beta(t) + u_\beta(t)$;} \nonumber \\
& \quad \bullet \text{ from state $R$ into state $S$ is $1$.} \nonumber
\end{align}
The transition rate from state $S$ into state $I$ accounts for the fact that the probability of being infected is directly proportional to the concentration of infected agents.

The transition rates provided above imply that the transient probabilities of the CTMC satisfy the Kolmogorov equations
\begin{align}\label{eq_sir_atomic_u}
\dot{\pi}^{u_\beta}_S & = - V^{u_\beta}_I \pi^{u_\beta}_S + \pi^{u_\beta}_R \\
\dot{\pi}^{u_\beta}_I & = - (\hk_\beta + u_\beta) \pi^{u_\beta}_I + V^{u_\beta}_I \pi^{u_\beta}_S \nonumber \\
\dot{\pi}^{u_\beta}_R & = - \pi^{u_\beta}_R + (\hk_\beta + u_\beta) \pi^{u_\beta}_I , \nonumber
\end{align}
where $\pi^{u_\beta}_S(t)$, $\pi^{u_\beta}_I(t)$ and $\pi^{u_\beta}_R(t)$ denotes the probability that the fixed agent is susceptible, infected and recovered at time $t$, respectively.

We now make the pivotal observation that the solution $\pi^{u_\beta}$ of~(\ref{eq_sir_atomic_u}) with the initial condition given by $\pi^{u_\beta}(0) = V^{u_\beta}(0)$ is also a solution of~(\ref{eq_nut1}), i.e., $\pi^{u_\beta} \equiv V^{u_\beta}$. (To see this, replace each $\pi^{u_\beta}_A$ with $V_A^{u_\beta}$ in~(\ref{eq_sir_atomic_u}).) Hence, if we are given $V^{u_\beta}_I$, the nonlinear ODE system~(\ref{eq_nut1}) can be expressed in terms of the linear Kolmogorov equations~(\ref{eq_sir_atomic_u}).

Unfortunately, we cannot use~(\ref{eq_sir_atomic_u}) directly to estimate $\calE$ from~(\ref{eq_eps_t_0}) because of the term $V^{u_\beta}_I$. We tackle this problem by replacing $V^{u_\beta}_I$ by $V^0_I + u_I$, where $V^0$ is the nominal trajectory of~(\ref{eq_nut1}) in the case of $u_\beta \equiv 0$ and $u_I$ is a new uncertainty function that satisfies $|u_I(\cdot)| \leq \eps_I(\cdot)$ for some positive function $\eps_I$. This yields the linear ODE system
\begin{align}\label{ex_eq_sir_atomic_uu}
\dot{\pi}^{u_\beta,u_I}_S & = - (V^0_I + u_I) \pi^{u_\beta,u_I}_S + \pi^{u_\beta,u_I}_R \\
\dot{\pi}^{u_\beta,u_I}_I & = - (\hk_\beta + u_\beta) \pi^{u_\beta,u_I}_I + (V^0_I + u_I) \pi^{u_\beta,u_I}_S \nonumber \\
\dot{\pi}^{u_\beta,u_I}_R & = - \pi^{u_\beta,u_I}_R + (\hk_\beta + u_\beta) \pi^{u_\beta,u_I}_I \nonumber
\end{align}

The key observation is that for any $u_\beta$, the uncertainty $u_I := V^{u_\beta}_I - V^0_I$ induces a solution of~(\ref{ex_eq_sir_atomic_uu}) which coincides with the solution of~(\ref{eq_sir_atomic_u}), meaning that $\pi^{u_\beta,u_I} \equiv V^{u_\beta}$ whenever $\pi^{u_\beta,u_I}(0) = V^{u_\beta}(0)$.

Moreover,~(\ref{ex_eq_sir_atomic_uu}) is a linear ODE system that is \emph{decoupled} from~(\ref{eq_nut1}). Thus, instead of considering the maximal deviation of $V^{u_\beta}$ from $V^0$, $\calE$, the above discussion motivates to focus on the maximal deviation of $\pi^{u_\beta,u_I}$ from $\pi^{0,0}$, i.e.,
\begin{multline*}
(\Phi_B(\eps))(t) = \sup \{ |\pi_B^{u_\beta,u_I}(t) - \pi_B^{0,0}(t)| \\
|u_\beta(\cdot)| \leq \delta_\beta(\cdot) \text{ and } |u_I(\cdot)| \leq \eps_I(\cdot) \} ,
\end{multline*}
where $\eps = (\eps_S(\cdot),\eps_I(\cdot),\eps_R(\cdot))$ is a positive function, $B \in \{S,I,R\}$ and $\pi^{0,0}$ denotes the nominal solution of~(\ref{ex_eq_sir_atomic_uu}) when $u_\beta \equiv 0$ and $u_I \equiv 0$.

%Intuitively, the function $\Phi$ takes as input a function $\eps$ and maps it to a new function $\eps' = (\eps'_S(\cdot),\eps'_I(\cdot),\eps'_R(\cdot))$. Our goal is to find a function $\eps$ that satisfies $\calE_B(t) \leq \eps_B(t)$ for all $B \in \{S,I,R\}$ and $0 \leq t \leq T$ (or $\calE \leq \eps$ for short). Since the input $\eps_I$ is, essentially, the guess for the maximal deviation of $\pi_I^{u_\beta,u_I}$ from $\pi_I^{0,0}$, a natural candidate is any function $\eps$ satisfying $\Phi(\eps) \leq \eps$. This is because $\Phi(\eps) \leq \eps$ implies that the maximal deviation of $\pi_I^{u_\beta,u_I}$ from $\pi_I^{0,0}$ is bounded by $\eps_I$ whenever $u_I$ itself is bounded by $\eps_I$. Building on this intuition, we will prove that $\Phi(\eps) \leq \eps$ indeed implies $\calE(T) \leq \eps$ and provide an algorithm that computes, whenever possible, the smallest such function $\eps$.

Intuitively, $\Phi$ takes a guess $\eps = (\eps_S,\eps_I,\eps_R)$ for $\calE$ as input and provides the new guess $\Phi(\eps) = (\Phi_S(\eps),\Phi_I(\eps),\Phi_R(\eps))$ for $\calE$. Our goal is to find an $\eps$ that satisfies $(\Phi_B(\eps))(t) \leq \eps_B(t)$ for all $B \in \{S,I,R\}$ and $0 \leq t \leq T$ (or $\Phi(\eps) \leq \eps$ for short). This is because $\eps_I$ is a guess for a bound on $|V_I^{u_\beta} - V_I^{0}|$, while $\Phi(\eps) \leq \eps$ implies that $|\pi_I^{u_\beta,u_I} - \pi_I^{0,0}|$ is bounded by $\eps_I$ whenever $|u_I|$ itself is bounded by $\eps_I$. Building on this intuition, we will prove that $\Phi(\eps) \leq \eps$ implies $\calE \leq \eps$ and provide an algorithm that computes, whenever possible, the smallest such positive function $\eps$.

\emph{\textbf{Second Step: Approximation of Kolmogorov equations.}} The above discussion shows that an estimation of $\calE$ requires one to evaluate the function $\Phi$. Thanks to the fact that~(\ref{ex_eq_sir_atomic_uu}) arises from~(\ref{eq_sir_atomic_u}) by decoupling, it can be seen that~(\ref{ex_eq_sir_atomic_uu}) describes the Kolmogorov equations of a CTMC with time-varying uncertain transition rates that are not coupled to the ODE system~(\ref{eq_nut1}). This, in turn, allows one to compute any value of $\Phi$ by solving a family of tractable optimization problems. More specifically, the value $(\Phi_B(\eps))(\tha)$ can be computed by determining two uncertainty functions $u^\ast_\beta$ and $u^\ast_I$ such that
\begin{multline}\label{eq_kolmogorov}
\pi_B^{u^\ast_\beta,u^\ast_I}(\tha) = \texttt{opt}\{ \pi_B^{u_\beta,u_I}(\tha) \mid \pi^{u_\beta,u_I} \text{ solves~(\ref{ex_eq_sir_atomic_uu}) and} \\
|u_\beta(\cdot)| \leq \delta_\beta(\cdot), |u_I(\cdot)| \leq \eps_I(\cdot) \} ,
\end{multline}
with $\pi^{u_\beta,u_I}(0) = V^{u_\beta}(0)$ and $\texttt{opt} \in \{\inf, \sup\}$. By interpreting the uncertainty functions $u^\ast_\beta$ and $u^\ast_I$ as optimal controls,~(\ref{eq_kolmogorov}) defines an optimal control problem with cost $\pi_B^{u_\beta,u_I}(\tha)$.

A major result of the paper shows that uncertainty functions $u^\ast_\beta$ and $u^\ast_I$ can be efficiently computed by relying on a strict version of Pontryagin's principle which is sufficient for optimality. Apart from solving~(\ref{eq_kolmogorov}) exactly, this allows one to devise an efficient procedure for the formal estimation of $\calE$ whose bounds can be expected to be tight.

\section{Technical Preliminaries}\label{sec_an}

In this section, we introduce agent networks (ANs), a class of nonlinear ODE systems to which our over-approximation technique can be applied. ANs are, essentially, chemical reactions networks whose reaction rate functions are not restricted to the law of mass action. The distinctive feature of ANs is that their dynamics can be related to the linear Kolmogorov equations of CTMCs.

\begin{definition}\label{def_agent_network}
An agent network (AN) is a triple $(\cals,\calK,\calF)$ of a finite set of states $\cals = \{A_1,\ldots,A_{|\cals|}\}$, parameters $\calK$ and reaction rate functions $\calF$. Each reaction rate function $\Theta_j : \RE_{ \geq 0}^{\cals \cup \calK} \rightarrow \REz$
\begin{itemize}
    \item describes the rate at which reaction $j$ occurs;
    \item takes concentration and parameter vectors $V \in \RE_{>0}^\cals$ and $\kappa \in \RE_{>0}^\calK$, respectively;
    \item is accompanied by a multiset $R_j$ of atomic transitions of the form $A_l \to A_{l'}$, where $A_l \to A_{l'}$ describes an agent in state $A_l$ interacting and changing state to $A_{l'}$.
\end{itemize}
\end{definition}

From a multiset $R_j$, we can extract two integer valued $|\cals|$-vectors $d_j$ and $c_j$, counting how many agents in each state are transformed during a reaction (respectively produced and consumed). Specifically, for each $1 \leq j \leq |\calF|$, let $c_{jl}, d_{jl} \in \mathbb{N}_0$ be such that
\[
c_{j,l} = \sum_{A_l \to A_{l'} \in R_j} 1 \quad \text{and} \quad d_{j,l'}  = \sum_{A_l \to A_{l'} \in R_j} 1.
\]
With these vectors, we can express the $j$-th reaction in the chemical reaction style~\cite{kurtz-chem} as follows:
\begin{align}\label{eq_gen_reaction}
c_{j,1} A_1 + \ldots + c_{j,|\cals|} A_{|\cals|} \act{\Theta_j}  d_{j,1} A_1 + \ldots + d_{j,|\cals|} A_{|\cals|}
\end{align}

We next introduce the ODE semantics of an AN.

\begin{definition}\label{def_global_ode_with_u}
For a given AN $(\cals,\calK,\calF)$, a continuous parameter function $\hk : [0;T] \to \RE_{>0}^\calK$ and a piecewise continuous function $\delta : [0;T] \to \RE^\calK_{>0}$ with $\delta_\alpha(\cdot) <  \hk_\alpha(\cdot)$ with $\alpha \in \calK$, let
\begin{multline*}
\calU_{\calK}^\delta := \{ u : [0;T] \to \RE_{>0}^\calK \mid |u_\alpha(\cdot)| \leq \delta_\alpha(\cdot) \\
\text{and } u \text{ is measurable} \}
\end{multline*}
denote the set of admissible uncertainties. Then, the reachable set of $(\cals,\calK,\calF)$ with respect to $\calU_\calK^\delta$ is given by the solution set $\{V^u \mid u \in \calU_{\calK}^\delta\}$, where $V^u$ solves
\begin{align}\label{eq_global_ode}
\dot{V}^u_B(t) & = F_B(V^u(t),\hk(t) + u(t))  \\
& := \sum_{1 \leq j \leq |\calF|} (d_{j,B} - c_{j,B}) \Theta_j\big(V^u(t), \hk(t) + u(t) \big) \nonumber
\end{align}
for all $B \in \cals$. The reachable set at time $0 \leq t \leq T$ is given by $\calR(t) = \{V^u(t) \mid u \in \calU_{\calK}^\delta\}$.
\end{definition}

The following example demonstrates Definition~\ref{def_agent_network} and~\ref{def_global_ode_with_u} in the context of the SIRS model from Section~\ref{sec_nut}. In particular, we remark the following.

\begin{remark}
Throughout the paper, the SIRS model from Section~\ref{sec_nut} is used to explain definitions and statements. All example environments refer to it.
\end{remark}

\begin{example}\label{ex_sir}
Consider the agent network $(\{S,I,R\},$ $\{\beta\},$ $\{\Theta_1,\Theta_2,\Theta_3\})$ given by
\begin{align*}
& R_1 \! = \! \{ S \to I, I \to I \},   & &   R_2 \! = \! \{ I \to R \},                     & & R_3 \! = \! \{ R \to S \}, \\
& \Theta_1(V,\kappa) \! = \! V_S V_I,   & &   \Theta_2(V,\kappa) \! = \! \kappa_\beta V_I,   & & \Theta_3(V,\kappa) \! = \! V_R ,
\end{align*}
where $V = (V_S,V_I,V_R)$ and $\kappa = (\kappa_\beta)$. Let the time-varying uncertain recovery rate parameter be given by $\kappa_\beta \equiv \hk_\beta + u_\beta$, where $\hk_\beta$ denotes the nominal trajectory and $u = (u_\beta) \in \calU^\delta_{\{\beta\}}$ is the uncertainty function for some positive $\delta = (\delta_\beta)$ such that $\delta_\beta < \hk_\beta$. The AN induces the reactions
\begin{align}\label{ex_sir_20}
S + I & \act{V_S V_I} I + I, & I & \act{(\hk_\beta + u_\beta) V_I} R, &  R & \act{V_R} S ,
\end{align}
while the ODE system~(\ref{eq_global_ode}) is given by~(\ref{eq_nut1}).
\end{example}

In the following, we assume that an AN $(\cals,\calK,\calF)$ is accompanied by a finite time horizon $T > 0$, a positive initial condition $V(0) \in \RE^\cals_{>0}$ and a Lipschitz continuous parameter function $\hk \in [0;T] \to \RE_{>0}^\calK$. Moreover, we require that each function $\Theta_j$
\begin{enumerate}[$i)$]
    \item is analytic in $(V,\kappa)$ and linear in $\kappa$, i.e., it holds that $\Theta_j(V, c \kappa + c' \kappa') = c \Theta_j(V, \kappa) + c' \Theta_j(V, \kappa')$;
    \item satisfies $\Theta_j(V) = 0$ whenever $V_{A_l} = 0$ and $c_{j,l} > 0$, where $c_{j,l}$ is as in~(\ref{eq_gen_reaction}).
\end{enumerate}

%The constraint $\delta < \min_{\alpha \in \calK} \min_{0 \leq t \leq T} \hk_\alpha(t)$ from Definition~\ref{eq_global_ode_with_u} ensure that the ODE solution~(\ref{eq_global_ode}) is positive.

Condition $i)$ enforces the existence of a unique solution~(\ref{eq_global_ode}) and allows us to apply Pontryagin's principle in Section~\ref{sec_mdps}, while condition $ii)$ says essentially that the $j$-th reaction~(\ref{eq_gen_reaction}) can only take place when all its reactants have a positive concentration. % and ensures that there is an $\eta > 0$ such that $V^u(t) \in \RE^\cals_{\geq \eta}$ for all $u \in \calU^\delta_\calK$ and $0 \leq t \leq T$.

We wish to point out that $i)$ and $ii)$ can be easily checked because analytic functions are closed under summation, multiplication and composition. Additionally, functions $\Theta_j$ often enjoy a simple form in practical models (in the case of biochemistry, for instance, they are given in terms of monomials).

With $i)$ and $ii)$ in place, the following can be proven.

\begin{proposition}\label{prop_existence}
In the case $i)$ and $ii)$ hold true,~(\ref{eq_global_ode}) admits a unique solution $V^u$ on $[0;T]$ for any uncertainty function $u \in \calU^\delta_\calK$. Moreover, there exists an $\eta > 0$ such that $V^u(t) \in \RE^\cals_{\geq \eta}$ for all $u \in \calU^\delta_\calK$ and $0 \leq t \leq T$.
\end{proposition}

\begin{proof}
Local existence and uniqueness are ensured by~\cite[Section 3.3.1]{Liberzon}. Let us define $W(0) := V(0)$,
\begin{align*}%\label{eq_global_ode_lower}
\dot{W}_B(t) & \! = \! G_B(t,W(t)) \! := - \! \sum_{1 \leq j \leq |\calF|} \! c_{j,B}\Theta_j\big(W(t), \hk(t) \! + \! \delta(t) \big)
\end{align*}
for all $B \in \cals$ and let $e(\alpha) \in \RE^\calK$ denote the vector with $e(\alpha)_{\alpha'} = 1$ if $\alpha = \alpha'$ and $e(\alpha)_{\alpha'} = 0$ when $\alpha \neq \alpha'$. With this, $ii)$ implies for all $1 \leq j \leq |\calF|$ and $u \in \calU^\delta_\calK$ that
%\begin{align*}%\label{eq_kappa_linearity}
%\Theta_j(V, \kappa) & = \Theta_j(V, (\kappa_\alpha)_{\alpha \in \calK}) = \sum_{\alpha \in \calK} \kappa_\alpha \Theta_j(V, e(\alpha) )
%\end{align*}
%This, in turn, yields for any $u \in \calU^\delta_\calK$
\begin{align*}
\Theta_j\big(W, \hk + u \big) & \leq \Theta_j(W, \hk) + \sum_{\alpha \in \calK} \delta_\alpha \Theta_j(W, e(\alpha))
\end{align*}
because the function $\Theta_j$ is nonnegative. Hence, $\dot{V}^u(t) = F(V^u(t),\hk(t) + u(t)) \geq G(t,V^u(t))$, thus implying that $V^u \geq W$ for all $u \in \calU^\delta_\calK$. We next show that $W$ is positive on $[0;T]$. To this end, let us assume towards a contradiction that there is $0 < \tau \leq T$ such that $W_A(\tau) = 0$ for some $A \in \cals$. Thanks to the continuity of $W$, we may assume without loss of generality that $W$ is positive on $[0;\tau)$. With $\calW(s) := W(\tau - s)$, it holds that $\dot{\calW}(s) = - G(\tau-s, \calW(s))$. There exists a sufficiently small interval $[0;\tau')$ on which  Euler's sequence given by $(\calW^{l})_{l\geq0}$, where $\calW^{0} := \calW(0)$ and $\calW^{l+1} := \calW^{l} - \dt \cdot G(\tau - l \dt, \calW^{l}) $, converges to a local solution of $\calW$~\cite{Gear:1971}. By construction, the sequence has to converge to a positive function on $(0;\tau')$ as $\dt \to 0$. However, thanks to $ii)$, $\calW_A^{0} = 0$ implies $\calW_A^{k} = 0$ for all $k \geq 0$ regardless how small $\dt > 0$ is, thus yielding a contradiction. Moreover, since $W > 0$ and $\sum_{B \in \cals} F_B(V^u(t),\hk(t) + u(t)) = 0$ for all $t \geq 0$, we also infer the existence of $V^u$ on the whole $[0;T]$.
\end{proof}

%Condition $ii)$ can be dropped whenever it can be proven that there is an $\eta > 0$ such that $V^u(t) \in \RE^\cals_{\geq \eta}$ for all $u \in \calU^\delta_\calK$ and $0 \leq t \leq T$.

It can be seen that atomic transitions enforce conservation of mass, i.e., the creation and destruction of agents is ruled out at the first sight. This problem, however, can be alleviated by the introduction of artificial agent states, see~\cite{DBLP:journals/iandc/BortolussiH15}.

\emph{Kolmogorov Equations of Agent Networks.} Thanks to the fact that the dynamics of an AN arise from atomic transitions, it is possible to define a CTMC underlying a given AN which Kolmogorov equations are closely connected to the ODE system~(\ref{eq_global_ode}).

\begin{definition}\label{def_transition_rate}
For a given AN $(\cals,\calK,\calF)$, define
\[
r_{B,C}(V,\kappa) = \sum_{1 \leq j \leq |\calF| \ \mid \ B \rightarrow C \, \in \, R_j} \Theta_j(V,\kappa) / V_B
\]
for all $B,C \in \cals$ with $B \neq C$, $V \in \RE^\cals_{>0}$ and $\kappa \in \RE_{>0}^\calK$. Then, the coupled CTMC $(X^u(t))_{t\geq0}$ underlying $(\cals,\calK,\calF)$ and $u \in \calU_{\calK}^\delta$ has state space $\cals$ and its transition rate from state $B$ into state $C$ at time $t$ is $r_{B,C}(V^u(t),\hk(t) + u(t))$. The coupled Kolmogorov equations of $(X^u(t))_{t\geq0}$ are
\begin{align}\label{eq_atomic_ode_with_u}
\dot{\pi}^u_B(t) & = f_B\big(\pi^u(t),V^u(t),\hk(t) + u(t)\big) \\
& := - \sum_{C : C \neq B} r_{B,C}(V^u(t),\hk(t) + u(t)) \pi^u_B(t) \nonumber \\
& \qquad + \sum_{C : C \neq B} r_{C,B}(V^u(t),\hk(t) + u(t)) \pi^u_C(t) \nonumber
\end{align}
\end{definition}
In the context of the SIRS example, Definition~\ref{def_transition_rate} gives rise to the transition rates~(\ref{eq_nut_trans}), the uncertainty function $u = (u_\beta)$ and the Kolmogorov equations~(\ref{eq_sir_atomic_u}). This is because the atomic transitions $S \to I$, $I \to R$ and $R \to S$ appear only in $R_1$, $R_2$ and $R_3$ of Example~\ref{ex_sir}, respectively, thus yielding
\begin{align*}
r_{S,I}(V^u(t),\hk(t) + u(t)) & = \Theta_1(V^u(t),\hk(t) + u(t)) / V^u_S(t) \\ %= V^u_I \\
r_{I,R}(V^u(t),\hk(t) + u(t)) & = \Theta_2(V^u(t),\hk(t) + u(t)) / V^u_I(t) \\ %= \hk_\beta + u_\beta \\
r_{R,S}(V^u(t),\hk(t) + u(t)) & = \Theta_3(V^u(t),\hk(t) + u(t)) / V^u_R(t) , %= 1
\end{align*}
where $\Theta_1, \Theta_2$ and $\Theta_3$ are as in Example~\ref{ex_sir}.

The next pivotal observation establishes a relation between the ODE system~(\ref{eq_global_ode}) and the Kolmogorov equations~(\ref{eq_atomic_ode_with_u}).

\begin{proposition}\label{prop_coinc_u}
For any uncertainty $u \in \calU_\calK^\delta$ and
\begin{align}\label{eq_init_pi}
\pi^u(0) = V(0) ,
\end{align}
the solution of~(\ref{eq_atomic_ode_with_u}) exists on $[0;T]$ and satisfies $\pi^u(t) = V^u(t)$ for all $0 \leq t \leq T$.
\end{proposition}

\begin{proof}
Note that~(\ref{eq_atomic_ode_with_u}) rewrites into~(\ref{eq_global_ode}) if $\pi^u_B$ and $\dot{\pi}^u_B$ is replaced with $V^u_B$ and $\dot{V}^u_B$ for all $B \in \cals$, respectively. With this, the claim follows from Proposition~\ref{prop_existence}. % by standard arguments from functional analysis and the theory of ODEs.
\end{proof}
In the context of the SIRS example, Proposition~\ref{prop_coinc_u} states that the solutions of~(\ref{eq_nut1}) and~(\ref{eq_sir_atomic_u}) coincide whenever $\pi^u(0) = V(0)$.

It is possible to prove that~(\ref{eq_global_ode}) and~(\ref{eq_atomic_ode_with_u}) are the fluid limits of certain CTMC sequences in the case where $\pi^u(0) = V(0) / \normo{V(0)}$ and the number of agents in the system tends to infinity, see~\cite{DBLP:journals/pe/BortolussiHLM13,DBLP:journals/iandc/BortolussiH15,Darling2008} for details. We will not elaborate on this relation further because it is not required for the understanding of our over-approximation technique.

\section{Over-Approximation Technique}\label{sec_ra}

As anticipated in Section~\ref{sec_nut} and~\ref{sec_an}, we estimate the reachable set of an AN with respect to an uncertainty set $\calU_{\calK}^\delta$, i.e., we bound $\calR(t) = \{V^u(t) \mid u \in \calU_{\calK}^\delta\}$ for each $0 \leq t \leq T$. To this end, we study the maximal deviation from the nominal trajectory $V^0$ attainable across $\calU_{\calK}^\delta$.

\begin{definition}\label{def_max_dev}
For a given AN $(\cals,\calK,\calF)$ with uncertainty set $\calU_{\calK}^\delta$, the maximal deviation at time $t$ of~(\ref{eq_global_ode}) from $V^0$ is
\begin{align}\label{eq_eps_t}
\mathcal{E}_B(t) = \sup_{u \in \calU_{\calK}^\delta} |V_B^u(t) - V_B^0(t)|
\end{align}
with $B \in \cals$ and $\calE = (\calE_B)_{B \in \cals}$. With this, it holds that
\[
\calR(t) \subseteq \prod_{B \in \cals} \big[ V_B^0(t) - \mathcal{E}_B(t); V_B^0(t) + \mathcal{E}_B(t) \big]
\]
\end{definition}

By Proposition~\ref{prop_coinc_u}, any trajectory $V^u$ of~(\ref{eq_global_ode}) coincides with the trajectory $\pi^u$ of~(\ref{eq_atomic_ode_with_u}) if $\pi^u(0) = V(0)$. Even though this allows one to relate the reachable set of a nonlinear system to that of a linear one, the transition rates of the coupled CTMC $(X^u(t))_{t\geq0}$ depend on $V^u$. We address this by decoupling the transition rates of the coupled CTMC from $V^u$.

\begin{definition}\label{def_atomic_ode_with_uu}
For $\eps < V^0$ and $\fru = (u_\calK,u_\cals) \in \calU_\calK^\delta \times \calU_\cals^\eps$, let $(\mathcal{D}^\fru(t))_{t\geq0}$ be the decoupled CTMC with transition rates $\big(r_{B,C}(V^0(t) + u_\cals(t),\hk(t) + u_\calK(t))\big)_{B,C}$ and the decoupled Kolmogorov equations
\begin{align}\label{eq_atomic_ode_with_uu}
\dot{\pi}^\fru(t) & = h\big(t,\pi^\fru(t),(u_\calK(t),u_\cals(t))\big) \\
& := f\big(\pi^\fru(t),V^0(t) + u_\cals(t),\hk(t) + u_\calK(t)\big) \nonumber ,
\end{align}
where $f$ is as in Definition~\ref{def_transition_rate} and $\calU_\cals^\eps$ is defined similarly to $\calU_\calK^\delta$ from Definition~\ref{def_global_ode_with_u}.
\end{definition}
In the context of the AN from Example~\ref{ex_sir}, the \emph{decoupled} Kolmogorov equations~(\ref{eq_atomic_ode_with_uu}) are given by~(\ref{ex_eq_sir_atomic_uu}) with $\fru \equiv (u_\calK,u_\cals) \equiv ((u_\beta),(u_I)) \in \calU_\calK^\delta \times \calU_\cals^\eps $ $ = \calU_{\{\beta\}}^\delta \times \calU_{\{S,I,R\}}^\eps$. This is because the transition rates of the \emph{decoupled} CTMC are
\begin{align}\label{eq_a1a2}
r_{S,I}(V^0(t) + u_\cals(t),\hk(t) + u_\calK(t)) & = V^0_I(t) + u_I(t) \\
r_{I,R}(V^0(t) + u_\cals(t),\hk(t) + u_\calK(t)) & = \hk_\beta(t) + u_\beta(t) \nonumber \\
r_{R,S}(V^0(t) + u_\cals(t),\hk(t) + u_\calK(t)) & = 1 \nonumber
\end{align}
A direct comparison with the transition rates of the \emph{coupled} CTMC given in~(\ref{eq_nut_trans}) reveals that the original transition rate from $S$ into $I$, $V_I^{u_\beta}(t)$, is replaced with $V_I^0(t) + u_I(t)$.

\begin{remark}
Note that $V^0$ can be efficiently computed using a numerical ODE solver and by setting $u$ in~(\ref{eq_global_ode}) to zero.
\end{remark}

The next result relates the original ODE system~(\ref{eq_global_ode}) to the decoupled Kolmogorov equations~(\ref{eq_atomic_ode_with_uu}).

\begin{proposition}\label{prop_overapprox}
Assume that $\mathcal{E} < V^0$. Then, for any $u_{\calK} \in \calU_{\calK}^\delta$, there exists some $u_\cals \in \calU_\cals^{\mathcal{E}}$ such that the solution of~(\ref{eq_atomic_ode_with_uu}) subject to the initial condition $V(0)$ satisfies $\pi^{\fru}(t) = V^{u_\calK}(t)$ for all $0 \leq t \leq T$.
\end{proposition}

\begin{proof}
For $\eps$ with $\mathcal{E} \leq \eps < V^0$, the definition of $\calE$ implies that $u_\cals := V^{u_\calK} - V^0 \in \calU_\cals^\eps$ for any $u_\calK \in \calU_\calK^\delta$. Since $\pi^{u_\calK,u_\cals}$ from~(\ref{eq_atomic_ode_with_uu}) coincides with $\pi^{u_\calK}$ from~(\ref{eq_atomic_ode_with_u}), Proposition~\ref{prop_coinc_u} yields the claim.
\end{proof}

To provide an estimation of $\mathcal{E}$ using the Kolmogorov equations~(\ref{eq_atomic_ode_with_uu}), we next define $\Phi(\eps)$ as the maximal deviation from the nominal trajectory $\pi^{0}$ that can be attained across the uncertainties $u_\calK \in \calU_\calK^\delta$ and $u_\cals \in \calU_\cals^\eps$.

\begin{definition}\label{def_phi}
For a piecewise continuous function $\eps < V^0$, let $\Phi(\eps) = (\Phi_B(\eps))_{B \in \cals}$ be given by
\begin{align*}
(\Phi_B(\eps))(t) = \sup_{u_\calK \in \calU_\calK^\delta} \sup_{u_\cals \in \calU_\cals^\eps} |\pi_B^\fru(t) - \pi_B^{0}(t)|
\end{align*}
$(\Phi_B(\eps))(t)$ denotes the maximal deviation of $\pi^\fru_B(t)$ from $\pi_B^{0}(t)$, where $\pi^0$ arises from $\pi^\fru$ in~(\ref{eq_atomic_ode_with_uu}) if $\fru = 0$.
\end{definition}

As discussed in Section~\ref{sec_nut}, the goal is to find a positive function $\eps$ such that $\Phi(\eps) \leq \eps$. This ensures that $|\pi^\fru - \pi^{0}| \leq \eps$ for any $\fru = (u_\calK,u_\cals) \in \calU_\calK^\delta \times \calU_\cals^\eps$ and implies, as stated in the next important result, that $\mathcal{E} \leq \eps$.

\begin{theorem}\label{thm_main_bound}
If $\Phi(\eps) \leq \eps$, then $\calE \leq \eps$.
\end{theorem}

\begin{remark}
For the benefit of presentation, we prove Theorem~\ref{thm_main_bound} in Section~\ref{sec_thm_main_bound} by invoking the strict version of Pontryagin's principle presented in Section~\ref{sec_mdps}.
\end{remark}

A direct consequence of Theorem~\ref{thm_main_bound} is that a fixed point $\eps^\ast$ of $\eps \mapsto \Phi(\eps)$ estimates $\mathcal{E}$ from above whenever $\eps^\ast < V^0$.

The next result describes an algorithm for the computation of the least fixed point $\eps^\ast$.

\begin{theorem}\label{thm_fp}
Fix some small $\eps^{(0)} > 0$ and set
\[
\eps^{(k + 1)} :=
\begin{cases}
\Phi(\eps^{(k)}) & , \ \eps^{(k)} < V^0 \\
\infty & , \ \text{otherwise}
\end{cases}
\]
for all $k \geq 0$. If $\lim_{k \to \infty} \eps^{(k)} = \eps$ such that $\eps \neq \infty$, then $\eps$ is the smallest fixed point of $\Phi$ which satisfies $\eps \geq \eps^{(0)}$.
\end{theorem}

\begin{proof}%[Proof of Theorem~\ref{thm_fp}]
Obviously, $\Phi$ is monotonic increasing, i.e., $\eps \leq \eps'$ implies $\Phi(\eps) \leq \Phi(\eps')$. With this, Kleene's fixed point theorem yields the claim.
\end{proof}

Note that the computation of the sequence $(\eps^{(k)})_k$ can be terminated if $\eps^{(k+1)} < V^0$ is violated because in such case no bound can be obtained.

\subsection{Optimal Solutions for inhomogeneous CTMDPs}\label{sec_mdps}

In each step of the fixed point iteration from Theorem~\ref{thm_fp}, a new value of $\Phi$ has to be computed. To this end, for any $0 \leq \tha \leq T$ and $A \in \cals$, we have to
\begin{multline}\label{eq_opt}
\text{obtain the minimal (maximal) value of $\pi_A(\tha)$ } \\
\text{ such that } \dot{\pi}(t) = h\big(t,\pi(t),(u_{\calK}(t),u_\cals(t))\big)  \\
\text{ subject to (\ref{eq_init_pi}) and $(u_{\calK},u_\cals) \in \calU_{\calK}^\delta \times \calU_\cals^\eps$}
\end{multline}
While the solution of such optimization problems is particulary challenging in the case of nonlinear dynamics, time-varying systems such as~(\ref{eq_opt}) are easier to come by. This is because~(\ref{eq_opt}) is a linear system with additive and multiplicative uncertainties. More formally,~(\ref{eq_opt}) is linear in concentrations variables if the parameter variables are fixed and linear in parameter variables when the concentration variables are fixed.

\begin{remark}\label{rem_ctmdp}
It is worth noting that~(\ref{eq_opt}) can be rewritten in the case of minimization (maximization is similar) to
\begin{multline}\label{eq_cost_func}
\min\{ \normo{V(0)} \cdot \mathbb{E}[\mathds{1}_{\mathcal{D}^\fru(\tha) = A}]  \mid  \\
\pi^\fru(0) = V(0) / \normo{V(0)} \text{ and } \fru \in \calU_{\calK}^\delta \times \calU_\cals^\eps \}
\end{multline}
This defines a CTMDP with finite state space $\cals$ and action space $\big( \prod_{\alpha \in \calK} [-\delta_\alpha(t); \delta_\alpha(t)] \big) \times \big( \prod_{A \in \cals} [-\eps_A(t); \eps_A(t)] \big)$ at time $t$. The CTMDP is inhomogeneous due to the presence of the function $V^0$ in the transition rates from Definition~\ref{def_atomic_ode_with_uu}.
\end{remark}

For the benefit of presentation, we write in that what follows $\fru \in \calU_{\calK \cup \cals}^{\frb}$ instead of $(u_{\calK},u_\cals) \in \calU_{\calK}^\delta \times \calU_\cals^\eps$, where $\frb_\alpha = \delta_\alpha$ and $\frb_A = \varepsilon_A$ for all $\alpha \in \calK$ and $A \in \cals$, respectively. Moreover, we recall that a solution of a differential inclusion $\dot{z} \in G(z)$ is any absolutely continuous function $z$ which satisfies $\dot{z} \in G(z)$ almost everywhere.

We solve~(\ref{eq_opt}) by modifying the strict version of \mbox{Pontryagin's} principle~\cite{KAMIEN} which is sufficient for optimality. Our modification of~\cite{KAMIEN} is less general than the original because it is stated for CTMCs but it makes weaker assumptions (the concavity of $\hat{H}$ is required on positive values only).

\begin{theorem}\label{thm_suff_pont}
For any $p \in \RE^\cals$, let $H(t,\pi,(u_{\calK},u_\cals),p) = \sum_{A \in \cals} p_A h_A(t,\pi,(u_{\calK},u_\cals))$ and assume that, for any $0 \leq t \leq \tha$ and $p \in \REz^\cals$, the function
\begin{multline*}
\pi \mapsto \hat{H}(t,\pi,p) = \max\big\{ H(t,\pi,(u_{\calK},u_\cals),p) \mid \\
u_\calK \in \prod_{\alpha \in \calK} [-\delta_\alpha(t);\delta_\alpha(t)], u_\cals \in \prod_{A \in \cals} [-\eps_A(t);\eps_A(t)] \big\}
\end{multline*}
is concave on $\RE_{>0}^\cals$. Then, any solution of the differential inclusion
\begin{align*}
\dot{\pi}(t) & \in h(t,\pi(t),u^\ast(t,\pi,p)) \\
\dot{p}(t) & \in - \sum_{B \in \cals} p_B (\partial_{\pi} h_B)(t,\pi(t),u^\ast(t))  \\
u^\ast(t) & \in \argmax_{(u_\calK, u_\cals)} H\big(t,\pi(t),(u_\calK,u_\cals),p(t)\big)
\end{align*}
subject to~(\ref{eq_init_pi}) and $p(\tha) \equiv - \mathds{1}_{\{A=\cdot\}}(\cdot)$ ($p(\tha) \equiv \mathds{1}_{\{A=\cdot\}}(\cdot)$) minimizes (maximizes) the value of $\pi_A(\tha)$.
\end{theorem}

\begin{proof}%[Proof of Theorem~\ref{thm_suff_pont}]
The proof follows the argumentation of~\cite{KAMIEN}. Fix some $u_\calK \in \calU_\calK^\delta \cap \mathcal{C}([0;\tha])$ and $u_\cals \in \calU_\cals^\eps \cap \mathcal{C}([0;\tha])$ and let $\pi$ denote the solution underlying $\dot{\pi}(t) = h(t,\pi(t),(u_{\calK}(t),u_\cals(t)))$. Note that it suffices to consider continuous uncertainties because standard results from ODE theory and functional analysis ensure that the maximal value $\Phi(\eps)$ can be attained by continuous uncertainties, that is
\[
(\Phi_B(\eps))(\tha) = \sup_{u_\calK \in \calC^\delta_\calK} \sup_{u_\cals \in \calC_\cals^\eps} | \pi_B^\fru(\tha) - \pi_B^{0}(\tha) | ,
\]
where $\calC^\delta_\calK = \calU^\delta_\calK \cap \mathcal{C}([0;\tha])$ and $\calC^\eps_\cals = \calU^\eps_\cals \cap \mathcal{C}([0;\tha])$. For the ease of notation, let $\pi^\ast$, $p$ and $u^\ast$ denote a solution of the differential inclusion and set
\begin{align*}
p \cdot h' & := p \cdot h(t,\pi,u^\ast) \\
p \cdot h^\ast & := p \cdot h(t,\pi^\ast,u^\ast) \\
\partial_\pi (p \cdot h^\ast) & := p \cdot (\partial_\pi h)(t,\pi^\ast,u^\ast) ,
\end{align*}
where $\cdot$ denotes the dot product. Thanks to the fact that $\dot{p}(t) = - p(t) \cdot (\partial_{\pi} h)(t,\pi^\ast(t),u^\ast(t))$, integration by parts yields
\begin{align*}
\int_{0}^{\tha} \dot{p} \cdot (\pi - \pi^\ast) dt = [ p \cdot (\pi - \pi^\ast) ]_{0}^{\tha} - \int_{0}^{\tha} p \cdot (h - h^\ast) dt
\end{align*}
With this, it holds that
\begin{align*}
0 & \geq \int_{0}^{\tha} \big( p \cdot h' - p \cdot h^\ast + \partial_\pi(p \cdot h^\ast) \cdot (\pi^\ast - \pi) \big) dt \\
& = \int_{0}^{\tha} \big( p \cdot h' - p \cdot h^\ast + \dot{p} \cdot (\pi - \pi^\ast) \big) dt \\
& = \int_{0}^{\tha} \big( p \cdot h' - p \cdot h^\ast - p \cdot h + p \cdot h^\ast \big) dt + [ p \cdot (\pi - \pi^\ast) ]_{0}^{\tha} \\
& \geq [ p \cdot (\pi - \pi^\ast) ]_{0}^{\tha} \\
& = p(\tha) \cdot (\pi(\tha) - \pi^\ast(\tha)) ,
\end{align*}
where the first inequality is implied by the concavity of $\pi \mapsto p \cdot h^\ast$, while the second inequality follows from the definition of $p \cdot h'$ and the choice of $u^\ast$. In the case where we seek to maximize the value of $\pi_A(\tha)$, we note that $p_\cdot(\tha) = \mathds{1}_{\{A=\cdot\}}$ yields $0 \geq \pi_A(\tha) - \pi^\ast_A(\tha)$. Since the case of minimization is similar, the proof is complete.
\end{proof}

We next identify structural conditions on $(\mathcal{D}^\fru(t))_{t\geq0}$ which can be easily checked and that imply the technical requirement of concavity of Theorem~\ref{thm_suff_pont}. %We will drop them in Section~\ref{sec_suboptimal}.
\begin{itemize}
    \item[\textbf{(A1)}] For any $B,C \in \cals$ and $0 \leq t \leq T$, there exist Lipschitz continuous $k^{B \to C}, k^{B \to C}_i \in [0;T] \to \REz$ such that
    the transition rate function $r_{B,C}$ from Definition~\ref{def_transition_rate} satisfies
    \begin{multline*}
    r_{B,C}\big(V^0(t) + u_\cals,\hk(t) + u_{\calK}\big) \\ = k^{B \to C}(t) + \sum_{i \in \calK \cup \cals} k^{B \to C}_i(t) u_i
    \end{multline*}
    for all $u_\calK \in \RE^\calK$ and $u_\cals \in \RE^\cals$.
    \item[\textbf{(A2)}] For each $i \in \calK \cup \cals$, there exist unique $B_i, C_i \in \cals$ such that $k_i^{B \to C} \not \equiv 0$ implies $B = B_i$, $C = C_i$ and $k_i^{B \to C} > 0$.
\end{itemize}
Assumption \textbf{(A1)} requires, essentially, the transition rate functions to be linear in the uncertainties, while \textbf{(A2)} forbids the same uncertainty to affect more than one transition of the decoupled CTMC $(\mathcal{D}^\fru(t))_{t\geq0}$.

The next example demonstrates that our running example satisfies condition \textbf{(A1)} and \textbf{(A2)}.

\begin{example}\label{ex_a1a2}
Recall that the transition rates of the decoupled CTMC of Example~\ref{ex_sir} are given by~(\ref{eq_a1a2}). Hence, $k^{S \to I} \equiv V^0_I$, $k^{I \to R} \equiv \hk_\beta$ and $k^{R \to S} \equiv 1$ and \textbf{(A1)} holds true. Condition \textbf{(A2)}, instead, follows with $B_I = S$, $C_I = I$, $k_I^{S \to I} \equiv 1$ and $B_\beta = I$, $C_\beta = R$, $k_\beta^{I \to R} \equiv 1$.
\end{example}

The following crucial theorem can be shown in the presence of $\textbf{(A1)}-\textbf{(A2)}$. We wish to stress that the result can be also applied to an ICTMDP which is not induced by an AN.
\begin{theorem}\label{thm_suff_cond_for_suff_pont}
Assume that $\textbf{(A1)}-\textbf{(A2)}$ hold true and fix some $A \in \cals$. Then, the differential inclusion
\begin{align}\label{eq_pont_cost}
\dot{p}_B(t) & \in \sum_{C \in \cals} (p_B(t) - p_C(t)) k^{B \to C}(t) \\
& \quad + \sum_{\substack{i \in \calK \cup \cals : \\ B_i = B}} (p_{B}(t) - p_{C_i}(t)) k_i^{B \to C_i}(t) u^\ast_i(t,p(t)) \nonumber
\end{align}
subject to $p_B(\tha) = - \mathds{1}_{\{A=B\}}$ ($p_B(\tha) = \mathds{1}_{\{A=B\}}$), with $B \in \cals$, $0 \leq t \leq \tha$ and
\begin{align}\label{eq_opt_control}
\psi_i(t,p(t)) & = \big(p_{C_i}(t) - p_{B_i}(t)\big) k_i^{B_i \to C_i}(t) , \nonumber \\
u^\ast_i(t,p(t)) & \in
\begin{cases}
\{ \frb_i(t) \} & , \ \psi_i(t,p(t)) > 0 \\
[-\frb_i(t) ; \frb_i(t)] & , \ \psi_i(t,p(t)) = 0 \\
\{ -\frb_i(t) \} & , \ \psi_i(t,p(t)) < 0 ,
\end{cases}
\end{align}
for $i \in \calK \cup \cals$, has a solution. Moreover, for any solution $p$ of~(\ref{eq_pont_cost}), (\ref{eq_opt_control}), the underlying solution $\pi$ that satisfies~(\ref{eq_init_pi}) and
\begin{align}\label{eq_pont_pi}
\dot{\pi}(t) & = h\big(t,\pi(t),u^\ast(t,p(t))\big) , \quad 0 \leq t \leq \tha
\end{align}
minimizes (maximizes) the value $\pi_A(\tha)$.
\end{theorem}

\begin{proof}%[Proof of Theorem~\ref{thm_suff_cond_for_suff_pont}]
In the following, we verify that a solution of~(\ref{eq_pont_cost})-(\ref{eq_pont_pi}) is a solution of the differential inclusion from Theorem~\ref{thm_suff_pont}. To this end, we first observe that
\begin{align*}
& H\big(t,\pi,(u_{\calK},u_\cals),p\big) = \sum_{B \in \cals} p_B h_B(t,\pi,(u_{\calK},u_\cals)) \\
& \quad = \sum_{B,C \in \cals} (p_C - p_B) \Big( k^{B \to C} + \sum_{i \in \calK \cup \cals} k_i^{B \to C} u_i \Big) \pi_B \\
& \quad = \sum_{B,C \in \cals} (p_C - p_B) k^{B \to C} \pi_B \\
& \qquad + \sum_{i \in \calK \cup \cals} (p_{C_i} - p_{B_i}) \pi_{B_i} k_i^{B_i \to C_i} u_i
\end{align*}
Hence, we infer that
\begin{align*}
& \max_{u_{\calK},u_\cals} H(t,\pi,(u_\calK,u_\cals),p) = \sum_{B,C \in \cals} (p_C - p_B) k^{B \to C} \pi_B \\
& \qquad + \sum_{i \in \calK \cup \cals} \max_{u_i} \big( (p_{C_i} - p_{B_i}) k^{B_i \to C_i}_i \pi_{B_i} \big) u_i
\end{align*}
This and Theorem~\ref{thm_suff_pont} show that an optimal control $u^\ast$ must satisfy~(\ref{eq_opt_control}). Moreover, it implies that
\begin{multline*}
\max_{u_{\calK},u_\cals} H(t,\lambda \pi + (1 - \lambda) \pi',(u_\calK,u_\cals),p) \\ = \lambda \max_{u_\calK,u_\cals} H(t,\pi,(u_\calK,u_\cals),p) \\ + (1 - \lambda) \max_{u_\calK,u_\cals} H(t,\pi',(u_\calK,u_\cals),p)
\end{multline*}
for all $0 \leq \lambda \leq 1$ and $\pi, \pi' \in \RE_{>0}^\cals$, thus yielding linearity (and thus also concavity) of $\hh$ on $\RE_{>0}^\cals$. The last statement follows by noting that
\begin{align*}
- \dot{p}_E & = \partial_{\pi_E} \Big( \sum_{B,C \in \cals} (p_C - p_B) k^{B \to C} \pi_B \\
& \quad + \sum_{i \in \calK \cup \cals} (p_{C_i} - p_{B_i}) k^{B_i \to C_i}_i \pi_{B_i} u_i \Big) \\
& = \sum_{C \in \cals} (p_C - p_E) k^{E \to C} \\
& \quad + \sum_{i : B_i = E} (p_{C_i} - p_E) k^{E \to C_i}_i u_i
\end{align*}
for all $E \in \cals$.
\end{proof}

We wish to stress that Theorem~\ref{thm_suff_cond_for_suff_pont} ensures that any solution $p$ of the differential inclusion~(\ref{eq_pont_cost}), (\ref{eq_opt_control}) induces an ODE solution $\pi$ of~(\ref{eq_pont_pi}) such that $\pi_A(\tha) = \pi_A^\ast(\tha)$, where $\pi^\ast_A(\tha)$ denotes the solution of~(\ref{eq_opt}). This stands in stark contrast to the standard version of Pontryagin's principle~\cite{Liberzon} which provides only necessary conditions for optimality, meaning that the value $\pi_A(\tha)$ arising from the standard version~\cite{Liberzon} may fail to satisfy $\pi_A(\tha) = \pi^\ast_A(\tha)$.

Solving a differential inclusion is a challenging task and requires one to assume in practice that it does not exhibit sliding or gazing modes~\cite{NumericDiffInc1,DBLP:conf/qest/Bortolussi11}. Fortunately, the next crucial results states that it is possible to obtain a specific solution of the differential inclusion~(\ref{eq_pont_cost}), (\ref{eq_opt_control}) by solving a Lipschitz continuous ODE system. %Since any solution of the differential inclusion~(\ref{eq_pont_cost}), (\ref{eq_opt_control}) induces a solution of~(\ref{eq_opt}), this allows one to solves optimization problems of the form~(\ref{eq_opt}) by solving two Lipschitz continuous ODE systems of size $|\cals|$.

\begin{theorem}\label{prop_lip}
By replacing~(\ref{eq_opt_control}) with
\begin{align}\label{eq_opt_control_det}
u^\ast_i(t,p(t)) =
\begin{cases}
\frb_i(t) & , \ \psi_i(t,p(t)) \geq 0  \\
-\frb_i(t) & , \ \psi_i(t,p(t)) < 0 \
\end{cases}
\end{align}
the differential inclusion~(\ref{eq_pont_cost}) becomes an ODE system which is Lipschitz continuous in $t$ and $p$. With this change in place, the statement of Theorem~\ref{thm_suff_cond_for_suff_pont} remains valid.
\end{theorem}

\begin{proof}%[Proof of Theorem~\ref{prop_lip}]
Let $\calP$ denote the drift of the ODE system~(\ref{eq_pont_cost}) which underlies~(\ref{eq_opt_control_det}), that is
\begin{multline*}
\calP_B(t,p) = \sum_{C \in \cals} (p_B - p_C) k^{B \to C}(t) \\
+ \sum_{i \in \calK \cup \cals} (p_{B_i} - p_{C_i}) k_i^{B_i \to C_i}(t) u^\ast_i(t,p) ,
\end{multline*}
with $u^\ast$ being as in~(\ref{eq_opt_control_det}). Fix some $(\tilde{t},\tilde{p}) \in [0;\tha] \times \RE^\cals$ and pick further two sequences $(t^l,p^l)_l$ and $(\tau^l,\wp^l)_l$ in $[0;\tha] \times \RE^\cals$ which converge both to $(\tilde{t},\tilde{p})$ as $l \to \infty$. We first show that $\calP_B(t^l,p^l) - \calP_B(\tau^l,\wp^l) \to 0$ as $l \to \infty$. To this end, it suffices to observe that any $i \in \calK \cup \cals$ with $\psi_i(\tilde{t},\tilde{p}) = 0$ implies $\tilde{p}_{B_i} - \tilde{p}_{C_i} = 0$ (recall that $k_i^{B \to C} \not \equiv 0$ yields $k_i^{B \to C} > 0$). Hence, it holds that
\begin{align*}
& | (p^l_{B_i} - p^l_{C_i}) k_i^{B_i \to C_i}(t^l) u^\ast_i(t^l,p^l) - \\
& \qquad (\wp^l_{B_i} - \wp^l_{C_i}) k_i^{B_i \to C_i}(\tau^l) u^\ast_i(\tau^l,\wp^l) | \\
& \quad \leq \sup_{0 \leq t \leq \tha} \frb_i(t) k_i^{B_i \to C_i}(t) \big(| p^l_{B_i} - p^l_{C_i}| + |\wp^l_{B_i} - \wp^l_{C_i}|\big) \to 0
\end{align*}
as $l \to \infty$. This shows the continuity of $\calP$. To see also the Lipschitzianity, define
\begin{align*}
G^+_i & = \{ (t,p) \in [0;\tha] \times \RE^\cals \mid  p_{C_i} - p_{B_i} > 0 \} \\
G^-_i & = \{ (t,p) \in [0;\tha] \times \RE^\cals \mid  p_{C_i} - p_{B_i} < 0 \}
\end{align*}
for each $i \in \calK \cup \cals$ with $k_i^{B_i \to C_i} \not \equiv 0$. Note that $\psi_i(t,p) > 0$ if and only if $p_{C_i} - p_{B_i} > 0$ because $k_i^{B_i \to C_i} > 0$ whenever $k_i^{B_i \to C_i} \not \equiv 0$. Moreover, for any $s \in \{-1,+1\}^{\calK \cup \cals}$, $\calP$ is Lipschitz continuous on any bounded subset of $\bigcap_i G^{s_i}_i$ because $k_i^{B_i \to C_i}$ and $k^{B_i \to C_i}$ are Lipschitz continuous on $[0;\tha]$. This shows that $\calP$ is Lipschitz continuous on any bounded subset of $\bigcup_s \bigcap_i G^{s_i}_i$. With this, the continuity of $\calP$ implies that $\calP$ is Lipschitz continuous on any bounded subset of $[0;\tha] \times \RE^\cals$.
\end{proof}

In the remainder of the paper, we replace~(\ref{eq_opt_control}) by~(\ref{eq_opt_control_det}). Theorem~\ref{prop_lip} ensures that~(\ref{eq_pont_cost}) admits a unique solution $p$ and that the underlying optimal uncertainty $u^\ast(\cdot,p(\cdot))$ induces the minimal (maximal) value $\pi_A^\ast(\tha)$ via~(\ref{eq_opt_control_det}) and~(\ref{eq_pont_pi}).

\begin{figure}[tp!]
\centering
%\subfloat[]{%
%\includegraphics[width=0.18\textwidth]{img/sir_p.eps}
%}
%\ \
\subfloat{%
\includegraphics[width=0.18\textwidth]{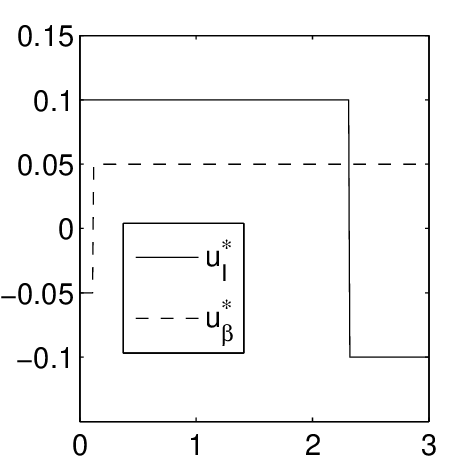}
}
\ \
\subfloat{%
\includegraphics[width=0.18\textwidth]{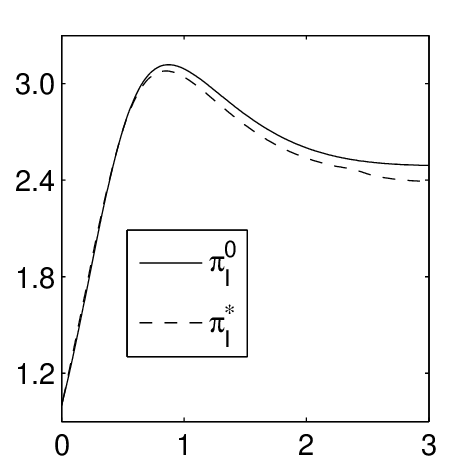}
}
%\ \
%\subfloat[]{%
%\includegraphics[width=0.18\textwidth]{img/sir_piImax.eps}
%}
\caption{Visualization of the optimal control $u^\ast$ and the underlying solution $\pi_I^\ast$ of~(\ref{eq_pont_pi}) whose value $\pi_I^\ast(3.0)$ solves the minimization problem~(\ref{eq_opt}) in the case of $\delta_\beta \equiv 0.05$, $\eps_I \equiv 0.10$, $\pi^0_S(0) = 4$, $\pi^0_I(0) = 1$ and $\pi^0_R(0) = 1$.}\label{fig_sir}
\end{figure}

The next example demonstrates Theorem~\ref{thm_suff_cond_for_suff_pont} and~\ref{prop_lip} in the context of the SIRS model from Example~\ref{ex_sir}.
\begin{example}\label{ex_sir2_3}
We have seen in Example~\ref{ex_a1a2} that our running example satisfies the requirements of Theorem~\ref{thm_suff_cond_for_suff_pont}. In particular, if $\hk \equiv 1$, then~(\ref{eq_pont_cost}) and~(\ref{eq_opt_control_det}) rewrite to
%\begin{example}\label{ex_costate}
\begin{align}\label{eq_costate}
\dot{p}_S(t) & = \big(V_I(t) + u^\ast_I(t,p(t))\big) (p_S(t) - p_I(t))  \\
\dot{p}_I(t) & = \big(\hk(t) + u^\ast_\beta(t,p(t))\big) (p_I(t) - p_R(t)) \nonumber \\
\dot{p}_R(t) & = p_R(t) - p_S(t) \nonumber
\end{align}
and
%\begin{align*}
%u^\ast_I(t,p(t)) & \in
%\begin{cases}
%\{ \eps \} & , \ p_I(t) - p_S(t) > 0 \\
%[- \eps ; \eps ] & , \ p_I(t) - p_S(t) = 0 \\
%\{ -\eps \} & , \ p_I(t) - p_S(t) < 0
%\end{cases} \\
%u^\ast_\beta(t,p(t)) & \in
%\begin{cases}
%\{ 0.05 \} & , \ p_R(t) - p_I(t) > 0 \\
%[-0.05 ; 0.05] & , \ p_R(t) - p_I(t) = 0 \\
%\{ -0.05 \} & , \ p_R(t) - p_I(t) < 0
%\end{cases}
%\end{align*}
%By Theorem~\ref{prop_lip}, we can set $u^\ast_I(t,p(t))$ and $u^\ast_\beta(t,p(t))$ to
\begin{align*}
u^\ast_I(t,p(t)) & =
\begin{cases}
\eps_I(t) & , \ p_I(t) - p_S(t) \geq 0 \\
-\eps_I(t) & , \ p_I(t) - p_S(t) < 0
\end{cases} \\
u^\ast_\beta(t,p(t)) & =
\begin{cases}
\delta_\beta(t) & , \ p_R(t) - p_I(t) \geq 0 \\
-\delta_\beta(t) & , \ p_R(t) - p_I(t) < 0
\end{cases}
\end{align*}
respectively. The minimal value of, say, $\pi^\fru_I(\tha)$ can be obtained as follows. First, solve the ODE system~(\ref{eq_costate}) where the boundary condition is given by $p_I(\tha) = -1$ and $p_S(\tha) = p_R(\tha) = 0$. Afterwards, using the obtained solution $p$, solve the ODE system~(\ref{eq_pont_pi}) using the controls $u^\ast_I(\cdot,p(\cdot))$ and $u^\ast_\beta(\cdot,p(\cdot))$. A possible solution is visualized in Figure~\ref{fig_sir}.
\end{example}

While Theorem~\ref{prop_lip} solves the problem from a theoretical point of view, it has to be noted that a numerical solution $\underline{p}$ of the Lipschitz continuous ODE system~(\ref{eq_opt_control}), (\ref{eq_opt_control_det}) is an approximation of the true solution $p$. Hence, for any $\tilde{t}$ with $\psi_i(\tilde{t},p(\tilde{t})) \approx 0$, the computation of the optimal uncertainty $u_i(\tilde{t})$ may be hindered by the numerical errors underlying the ODE solver. The next crucial result addresses this issue by stating that, essentially, for each such $\tilde{t}$ the choice of $u_i(\tilde{t})$ is not important.

%which can be solved efficiently using numerical ODE solvers~\cite{Gear:1971}.

\begin{theorem}\label{thm_numeric}
For any $\xi > 0$, it is possible to efficiently compute some $\zeta > 0$ such that the following holds. If $\fru \in \calU_{\calK \cup \cals}^\frb$ is such that for all $i \in \calK \cup \cals$ it holds that $\fru_i(t) = u^\ast_i(t,p(t))$ whenever $|\psi_i(t,p(t))| \geq \zeta$, then $|\pi_A^\fru(\tha) - \pi_A^\ast(\tha)| \leq \xi$, where $\pi_A^\fru$ and $\pi_A^\ast$  is the solution of~(\ref{eq_atomic_ode_with_uu}) and~(\ref{eq_opt}), respectively.
\end{theorem}

\begin{proof}%[Proof of Theorem~\ref{thm_numeric}]
With $c_1 := 2 \sup \{\frb_i(t) \mid i \in \calK \cup \cals, 0 \leq t \leq T\}$, $c_2 := \normo{V(0)} \cdot \max \{ | k^{B_i \to C_i}_i(t) | \mid i \in \calK \cup \cals, 0 \leq t \leq T \}$ and $\zeta = \xi / (T \cdot |\calK \cup \cals| c_1 c_2)$, let the function $u \in \calU_{\calK \cup \cals}^\frb$ be such that $u_i(t) = u_i^\ast(t,p(t))$ whenever $|p_{C_i}(t) - p_{B_i}(t)| \geq \zeta$. Using the same notation as in the proof of Theorem~\ref{thm_suff_pont}, we infer in the case when $p(\tha) \equiv \mathds{1}_{\{A = \cdot\}}(\cdot)$ the following:
\begin{align*}
0 & = \int_{0}^{\tha} \big( p \cdot h' - p \cdot h^\ast + \partial_\pi(p \cdot h^\ast) \cdot (\pi^\ast - \pi) \big) dt \\
& = \int_{0}^{\tha} \big( p \cdot h' - p \cdot h^\ast + \dot{p} \cdot (\pi - \pi^\ast) \big) dt \\
& = \int_{0}^{\tha} \big( p \cdot h' - p \cdot h^\ast - p \cdot h + p \cdot h^\ast \big) dt + [ p \cdot (\pi - \pi^\ast) ]_{0}^{\tha} \\
& = \int_{0}^{\tha} \big( p \cdot h' - p \cdot h) dt + p(\tha) \cdot (\pi(\tha) - \pi^\ast(\tha)) \\
& = \int_{0}^{\tha} \big( p \cdot h' - p \cdot h) dt + \pi_A(\tha) - \pi^\ast_A(\tha),
\end{align*}
where the first identity holds true because $\pi \mapsto p \cdot h^\ast$ is linear (see proof of Theorem~\ref{thm_suff_cond_for_suff_pont}), while the other identities follow as in the proof of Theorem~\ref{thm_suff_pont}. The above calculation yields
\begin{align*}
|\pi_A(\tha) - \pi^\ast_A(\tha)| & = \Big| \int_{0}^{\tha} \big( p \cdot h' - p \cdot h) dt \Big| \\
& = \Big| \int_0^{\tha} \sum_{i \in \calK \cup \cals} (p_{C_i}(t) - p_{B_i}(t)) \cdot \pi_{B_i}(t) \cdot \ldots \\
& \qquad \ldots \cdot k_i^{B_i \to C_i}(t) \cdot (u_i^\ast(t,p(t)) - u_i(t)) dt \Big|  \\
& \leq T |\calK \cup \cals| c_1 c_2 \zeta \\
& = \xi ,
\end{align*}
where the second equality follows from the proof of Theorem~\ref{thm_suff_cond_for_suff_pont}.
\end{proof}

The above theorem states, essentially, that the choice of $u_i(t)$ is irrelevant at all time points $t$ with $|\psi_i(t,p(t))| < \zeta$. Hence, an uncertainty $\fru$ that is induced by a numerical solution of~(\ref{eq_pont_cost}) can be used to solve~(\ref{eq_opt}). %At this point it is also worth mentioning that ODE solvers can be applied to ODE systems whose right-hand sides are discontinuous in time~\cite[Section 4.6.3]{Gear:1971}.

%\begin{example}\label{ex_costate_2}
%Overall, Theorem~\ref{prop_lip} and~\ref{thm_numeric} allow us to avoid the pitfalls associated with the theory of differential inclusions.

We end the section by mentioning the following generalization of Theorem~\ref{thm_suff_cond_for_suff_pont}.

%In particular, it eliminates a technical obstacle which often prevents the application of Pontryagin's principle~\cite{Liberzon} in the context of formal verification.

\begin{remark}
The statement of Theorem~\ref{thm_suff_cond_for_suff_pont} extends to the case where one seeks to minimize (maximize) the linear combination $\sum_{A \in \cals} \sigma_A \pi_A(\tha)$, where $\sigma \in \RE^\cals$. The corresponding boundary condition of~(\ref{eq_pont_cost}) is given by $p(\tha) = - \sigma$ ($p(\tha) = \sigma$).
\end{remark}

\subsection{Proof of Theorem~\ref{thm_main_bound}}\label{sec_thm_main_bound}

Armed with Theorem~\ref{thm_suff_cond_for_suff_pont} and Theorem~\ref{prop_lip}, we are now in a position to prove Theorem~\ref{thm_main_bound} under the assumption that the decoupled CTMC from Definition~\ref{def_atomic_ode_with_uu} satisfies \textbf{(A1)} and~\textbf{(A2)}. The assumption will be dropped in Section~\ref{sec_suboptimal}.

\begin{proof}[Proof of Theorem~\ref{thm_main_bound}]
Let us assume towards a contradiction that $\delta$ and $\eps$ are positive piecewise constant functions, that $\hk$ is analytic and that there exists an analytic uncertainty function $u_\calK \in \calU_\calK^\delta$, a time $0 < \tha \leq T$ and some $A \in \cals$ such that $|V^{u_\calK}_A(\tha) - V^0_A(\tha)| = \eps_A(\tha)$ and $|V^{u_\calK}_B(t) - V^0_B(t)| < \eps_B(t)$ for all $B \in \cals$ and $0 \leq t < \tha$. Since $\eps_A(\tha) > 0$, we may assume without loss of generality that $V_A^{u_\calK}(\tha) > V_A^0(\tha)$. With this, we consider the optimization problem
\begin{multline}\label{eq_opt_proof}
\text{compute the maximal value of $\pi^{(u'_{\calK},u'_\cals)}_A(\tha)$ } \\
\text{ such that } \dot{\pi}(t) = h\big(t,\pi(t),(u'_{\calK}(t),u'_\cals(t))\big)  \\
\text{ subject to (\ref{eq_init_pi}) and $(u'_{\calK},u'_\cals) \in \calU_{\calK}^\delta \times \calU_\cals^{\eps}$} ,
\end{multline}
where $h$ is as in~(\ref{eq_atomic_ode_with_uu}). Let $\pi^\ast_A(\tha)$ denote the solution of~(\ref{eq_opt_proof}) and set $\fru := (u_\calK,u_\cals)$ with $u_\cals := V^{u_\calK} - V^0$. Since $\eps_A(\tha) = |\pi_A^{\fru}(\tha) - \pi_A^0(\tha)| \leq |\pi_A^\ast(\tha) - \pi_A^0(\tha)| = (\Phi_A(\eps))(\tha) \leq \eps_A(\tha)$, we infer $\pi^\ast_A(\tha) = \pi_A^\fru(\tha)$. Hence, $\fru$ is an optimal control and Theorem~\ref{thm_suff_cond_for_suff_pont} and~\ref{prop_lip} imply that $u_i(t) = u_i^\ast(t)$ whenever $\psi_i(t,p(t)) \neq 0$, where $u_i^\ast(t)$ is as in~(\ref{eq_opt_control_det}), $i \in \calK \cup \cals$ and $p$ solves~(\ref{eq_pont_cost}) and~(\ref{eq_opt_control_det}). At the same time, the analyticity of $u_\calK$, $\hk$ and $(\Theta_j)_j$ implies that $u_\cals = V^{u_\calK} - V^0$ is analytic as well~\cite{DBLP:journals/iandc/BortolussiH15}. Since $\delta$ and $\eps$ are piecewise constant and none of the $u_i$ can be locally constant (otherwise the $u_i$ in question would be constant on the whole $[0;T]$ by the identity theorem), we infer that $\psi_i(\cdot,p(\cdot)) \equiv 0$ for all $i \in \calK \cup \cals$. Recall from the proof of Theorem~\ref{thm_numeric} that
\begin{align}\label{eq_thm_main_bound}
\pi^\ast_A(\tha) - \pi_A^\fru(\tha) = \int_0^{\tha} \big( p \cdot h' - p \cdot h) dt
\end{align}
and that the Hamiltonian $H(t'',\pi'',(u''_\calK,u''_\cals),p'')$ is invariant with respect to the value of $u''_i$ when $\psi_i(t'',p'') = 0$. This, the above discussion and~(\ref{eq_thm_main_bound}) imply that $\pi^\ast_A(\tha) = \pi_A^{(u'_{\calK},u'_\cals)}(\tha)$ for any uncertainty $(u'_\calK,u'_\cals)$. As this contradicts $V_A^{u_\calK}(\tha) > V_A^0(\tha)$, we infer the statement of the theorem in the case where $\delta$ and $\eps$ are piecewise constant and $u_\calK$ and $\hk$ are analytic. Thanks to the fact that analytic and piecewise constant functions are dense in set of bounded measurable functions on $[0;T]$, this suffices the claim.
\end{proof}

\subsection{Sub-Optimal Solutions for inhomogeneous CTMDPs}\label{sec_suboptimal}

It may happen that the decoupled CTMC from Definition~\ref{def_atomic_ode_with_uu} violates \textbf{(A1)} or \textbf{(A2)}. We next discuss a procedure which allows one to transform a CTMC violating \textbf{(A1)}-\textbf{(A2)} into one which satisfies \textbf{(A1)}-\textbf{(A2)}. We convey the main ideas using concrete examples.

The extension of Example~\ref{ex_sir} discussed next induces a decoupled CTMC which violates~\textbf{(A1)}.

\begin{example}\label{ex_sir_atomic_uu_2}
Consider the agent network $(\{S,I,R\},$ $\{\alpha,\beta\},$ $\{\Theta_1,\Theta_2,\Theta_3\})$ given by
\begin{align*}
& R_1 \! = \! \{ S \to I, I \to I \},   & &   R_2 \! = \! \{ I \to R \},                     & & R_3 \! = \! \{ R \to S \}, \\
& \Theta_1(V,\kappa) \! = \! \kappa_\alpha V_S V_I,   & &   \Theta_2(V,\kappa) \! = \! \kappa_\beta V_I,   & & \Theta_3(V,\kappa) \! = \! V_R ,
\end{align*}
where $V = (V_S,V_I,V_R)$ and $\kappa = (\kappa_\alpha, \kappa_\beta)$. Let the time-varying uncertain infection and recovery parameter functions be given by $\kappa_\alpha \equiv \hk_\alpha + u_\alpha$ and $\kappa_\beta \equiv \hk_\beta + u_\beta$, respectively, where $u = (u_\alpha, u_\beta) \in \calU^\delta_{\{\alpha,\beta\}}$ and $\delta = (\delta_\alpha,\delta_\beta)$. Then, the AN induces the reactions
\begin{align*}
S + I & \act{(\hk_\alpha + u_\alpha) V_S V_I} I + I, & I & \act{(\hk_\beta + u_\beta) V_I} R, &  R & \act{V_R} S
\end{align*}
%and the transition rates of the coupled CTMC from Definition~\ref{def_transition_rate} are
%\begin{align*}
%r_{S,I}(V^u,\hk + u) & = (\hk_\alpha + u_\alpha) V_I \\
%r_{I,R}(V^u,\hk + u) & = (\hk_\beta + u_\beta) \\
%r_{R,S}(V^u,\hk + u) & = 1
%\end{align*}
and the transition rates of the decoupled CTMC from Definition~\ref{def_atomic_ode_with_uu} are given by
\begin{align*}
r_{S,I}(V^0 + u_\cals,\hk + u_\calK) & = (\hk_\alpha + u_\alpha) (V^0_I + u_I) \\
r_{I,R}(V^0 + u_\cals,\hk + u_\calK) & = \hk_\beta + u_\beta \\
r_{R,S}(V^0 + u_\cals,\hk + u_\calK) & = 1
\end{align*}
Since
\[
(\hk_\alpha + u_\alpha) (V^0_I + u_I) = \hk_\alpha V^0_I + \hk_\alpha u_I + u_\alpha V^0_I + u_\alpha u_I
\]
leads to the nonlinear term $u_\alpha u_I$, the decoupled CTMC does not satisfy \textbf{(A1)}.
\end{example}
The idea is to substitute any nonlinear expression of uncertainties by a new uncertainty that bounds the original nonlinear expression. For instance, in the case of Example~\ref{ex_sir_atomic_uu_2}, we substitute $u_\alpha u_I$ with the new uncertainty $u_{\alpha|I}$ and set $\frb_{\alpha|I} := \frb_\alpha \frb_I$ because $|u_{\alpha}(\cdot) u_I(\cdot)| \leq \frb_\alpha(\cdot) \frb_I(\cdot)$.

This motivates the following concept.

\begin{definition}\label{def_envelope}
For $\eps < V^0$ a family of transition rates $(\hr_{B,C})_{B,C}$ is an envelope of the transition rates $(r_{B,C})_{B,C}$ from Definition~\ref{def_atomic_ode_with_uu} if there exist Lipschitz continuous functions $k^{B \to C}, k^{B \to C}_i \in [0;T] \to \REz$, an index set $\calI$ with $\calI \cap (\calK \cup \cals) = \emptyset$ and a piecewise continuous function $b : [0;T] \to \RE^{\calI}_{>0}$ such that for all $(u_{\calK},u_\cals) \in \calU_{\calK}^{\delta} \times \calU_\cals^\eps$ one can pick a $u_\calI \in \calU_\calI^{b}$ so that
\begin{multline*}
r_{B,C}\big(V^0(t) + u_\cals(t), \hk(t) + u_\calK(t)\big) \\
= \underbrace{k^{B \to C}(t) + \sum_{i \in \calK \cup \cals \cup \calI} k^{B \to C}_i(t) u_i(t)}_{\displaystyle \hr_{B,C}(t,u_\calK(t),u_\cals(t),u_\calI(t)) :=}
\end{multline*}
for all $0 \leq t \leq T$. %Further, we require $\hr_{B,C} \geq 0$ for all $\fru = (u_\calK,u_\cals,u_\calI) \in \calU_{\calK \cup \cals \cup \calI}^\frb$.
\end{definition}
A possible envelope of the decoupled CTMC from Example~\ref{ex_sir_atomic_uu_2} is given by $\hr_{I,R} := r_{I,R}$, $\hr_{R,S} := r_{R,S}$ and
\begin{align*}
& \hr_{S,I}(t,u_\calK(t),u_\cals(t),u_\calI(t)) \\
& \qquad := \hk_\alpha(t) V^0_I(t) + \hk_\alpha(t) u_I(t) + u_\alpha(t) V^0_I(t) + u_{\alpha|I}(t) ,
\end{align*}
with $\calI = \{\alpha|I\}$ and $\frb_{\alpha|I} := \frb_\alpha \frb_I$.

By construction, any envelope satisfies \textbf{(A1)}. It may however happen that an envelope does not satisfy \textbf{(A2)}. To see this on an example, we extend Example~\ref{ex_sir_atomic_uu_2} to the multi-class SIRS model~\cite{dsn16BortolussiGast} in which the overall population of agents is partitioned into classes, thus providing a better picture of the actual spread dynamics~\cite{dsn13IacobelliTribastone}.

\begin{example}\label{exex_multi_sir}
With $N \geq 2$ being the number of classes, the multi-class SIRS agent network is given by the atomic reactions
\begin{align*}
R^{\nu,\mu}_1 & = \{ S_\nu \to I_\nu, I_\mu \to I_\mu \}, & \Theta^{\nu,\mu}_1(V,\kappa) & = \kappa_{\alpha_{\nu,\mu}} V_{S_\nu} V_{I_\mu} , \nonumber \\
R^{\nu}_2 & = \{ I_\nu \to R_\nu \}, & \Theta^{\nu}_2(V,\kappa) & = \kappa_{\beta_\nu} V_{I_\nu} , \nonumber \\
R^{\nu}_3 & = \{ R_\nu \to S_\nu \}, & \Theta^{\nu}_3(V,\kappa) & = \kappa_{\gamma_\nu} V_{R_\nu} ,
\end{align*}
where $1 \leq \nu, \mu \leq N$. In the case where all rates are subject to uncertainty, the reactions are
\begin{align}\label{ex_multi_sir}
S_\nu + I_\mu & \act{(\hk_{\alpha_{\nu,\mu}} + u_{\alpha_{\nu,\mu}}) V_{S_\nu} V_{I_\mu}} I_\nu + I_\mu \\
I_\nu & \act{(\hk_{\beta_\nu} + u_{\beta_\nu}) V_{I_\nu}} R_\nu \nonumber \\
R_\nu & \act{(\hk_{\gamma_\nu} + u_{\gamma_\nu}) V_{R_\nu}} S_\nu \nonumber
\end{align}
The first reaction expresses the fact that a susceptible agent of class $\nu$ may be infected by an infected agent from class $\mu$. The transition rates of the decoupled CTMC are
\begin{align*}
r_{S_\nu,I_\nu}(V^0 + u_\cals,\hk + u_\calK) & = \sum_\mu (\hk_{\alpha_{\nu,\mu}} + u_{\alpha_{\nu,\mu}}) (V^0_{I_\mu} + u_{I_\mu})  \\
r_{I_\nu,R_\nu}(V^0 + u_\cals,\hk + u_\calK) & = \hk_{\beta_\nu} + u_{\beta_\nu} \\
r_{R_\nu,S_\nu}(V^0 + u_\cals,\hk + u_\calK) & = \hk_{\gamma_\nu} + u_{\gamma_\nu}
\end{align*}
The nonlinear terms $u_{\alpha_{\nu,\mu}} u_{I_\mu}$ prevent the decoupled CTMC to satisfy \textbf{(A1)}. Similarly to Example~\ref{ex_sir_atomic_uu_2}, we thus consider the envelope
\begin{align}\label{ex_sir_multi_envelope}
\hr_{S_\nu,I_\nu} & := \sum_\mu \big( \hk_{\alpha_{\nu,\mu}} V^0_{I_\mu} + \hk_{\alpha_{\nu,\mu}} u_{I_\mu} + u_{\alpha_{\nu,\mu}} V^0_{I_\mu} + u_{\alpha_{\nu,\mu}|I_\mu} \big) \nonumber \\
\hr_{I_\nu,R_\nu} & := \hk_{\beta_\nu} + u_{\beta_\nu} \nonumber \\
\hr_{R_\nu,S_\nu} & := \hk_{\gamma_\nu} + u_{\gamma_\nu}
\end{align}
with $\frb_{\alpha_{\nu,\mu}|I_\mu} := \frb_{\alpha_{\nu,\mu}} \frb_{I_\mu}$. Unfortunately, envelope~(\ref{ex_sir_multi_envelope}) violates \textbf{(A2)} because each $u_{I_\mu}$ is contained in $\hr_{S_1,I_1}$, \ldots, $\hr_{S_N,I_N}$. %Put different, each $u_{I_\mu}$ affects more than one transition rate.
\end{example}

We continue by observing that envelope~(\ref{ex_sir_multi_envelope}) can be transformed into a set of transition rates which satisfies~\textbf{(A1)} and~\textbf{(A2)}. Indeed, if we substitute in each $\hr_{S_\nu,I_\nu}$ from~(\ref{ex_sir_multi_envelope}) the uncertainty $u_{I_\mu}$ with $u_{I_{\nu,\mu}}$, the transition rates
\begin{align}\label{ex_sir_multi_rates}
\tr_{S_\nu,I_\nu} & := \sum_\mu \Big( \hk_{\alpha_{\nu,\mu}} V^0_{I_\mu} + \hk_{\alpha_{\nu,\mu}} u_{I_{\nu,\mu}} + u_{\alpha_{\nu,\mu}} V^0_{I_\mu} + u_{\alpha_{\nu,\mu}|I_\mu} \big) \nonumber \\
\tr_{I_\nu,R_\nu} & := \hk_{\beta_\nu} + u_{\beta_\nu} \nonumber \\
\tr_{R_\nu,S_\nu} & := \hk_{\gamma_\nu} + u_{\gamma_\nu}
\end{align}
define a CTMC which satisfies~\textbf{(A1)} and~\textbf{(A2)}. This is because every uncertainty function $u_i$, where $i \in \calK \cup \cals \cup \calI$ and
\begin{align*}
\calI & = \{ \alpha_{\nu,\mu}|I_\mu \mid 1 \leq \nu, \mu \leq N \} \cup \{ I_{\nu,\mu} \mid 1 \leq \nu, \mu \leq N \}
\end{align*}
with $\frb_{I_{\nu,\mu}} := \frb_{I_\mu}$ and $\frb_{\alpha_{\nu,\mu}|I_\mu} := \frb_{\alpha_{\nu,\mu}} \frb_{I_\mu}$, appears in exactly one transition rate $\tr_{B,C}$.

This above discussion motivates the following.

\begin{definition}\label{def_coars}
Assume that $(\hr_{B,C})_{B,C}$ is an envelope of $(r_{B,C})_{B,C}$ given by
\[
\hr_{B,C} = k^{B \to C} + \sum_{i \in \calK \cup \cals \cup \calI} k^{B \to C}_i u_i
\]
for all $B,C \in \cals$ and let $\calI_0 \dot \cup \calI_1 = \calK \cup \cals \cup \calI$ be such that $(\hr_{B,C})_{B,C}$ violates \textbf{(A2)} for each $i \in \calI_1$. Then, the transition rate $\tr_{B,C}$ arises from $\hr_{B,C}$ by substituting each occurrence of $u_i$ in $\hr_{B,C}$ with $u_{i | B \to C}$, where $i \in \calI_1$. By setting $\frb_{i | B \to C} := \frb_i$, the coarsening of $(\hr_{B,C})_{B,C}$ is given by $(\tr_{B,C})_{B,C}$.
\end{definition}

\begin{remark}
Note that~(\ref{ex_sir_multi_rates}) is, up to a renaming of indices, a coarsening of the envelope~(\ref{ex_sir_multi_envelope}). This can be seen by substituting each $u_{I_{\nu,\mu}}$ with $u_{I_{\mu} | S_\nu \to I_\nu}$.
\end{remark}

The next result states that the coarsening of an envelope of the decoupled CTMC from Definition~\ref{def_atomic_ode_with_uu} allows one to estimate $\calE$ from Definition~\ref{def_max_dev}.

\begin{theorem}\label{thm_ctmc_estimator}
Given the decoupled CTMC from Definition~\ref{def_atomic_ode_with_uu}, let us assume that $(\hr_{B,C})_{B,C}$ is an envelope for $(r_{B,C})_{B,C}$. Let further $(\tr_{B,C})_{B,C}$ denote the coarsening of $(\hr_{B,C})_{B,C}$ as given in Definition~\ref{def_coars} and let
\begin{align}\label{eq_tilde_pi}
\dot{\tilde{\pi}}^\frut(t) & = \tilde{h}(t,\tilde{\pi}^\frut(t), \frut(t)) \nonumber \\
& := - \sum_{C : C \neq B} \tr_{B,C}(t,\frut(t)) \tilde{\pi}^\frut_B(t) \nonumber \\
& \qquad + \sum_{C : C \neq B} \tr_{C,B}(t,\frut(t)) \tilde{\pi}^\frut_C(t)
\end{align}
denote the Kolmogorov equations underlying the coarsening $(\tr_{B,C})_{B,C}$. Set further $\tilde{\pi}^\fru(0) = V(0)$ and
\[
(\Psi_B(\eps))(t) = \sup \big\{ |\tilde{\pi}_B^\frut(t) - \tilde{\pi}_B^{0}(t)| \ \big| \ \frut \in \calU^\delta_\calK \times \calU^\eps_\cals \times \calU^b_\calI \big\}
\]
Then, applying the fixed point iteration algorithm of Theorem~\ref{thm_fp} to $\Psi$ instead of $\Phi$ yields a bound on $\calE$. Moreover, Theorem~\ref{thm_suff_cond_for_suff_pont} and Theorem~\ref{prop_lip} carry over to the extended set of uncertainties $\calU^\delta_\calK \times \calU^\eps_\cals \times \calU^b_\calI$ and can be used to compute $\Psi$.
\end{theorem}

\begin{algorithm}[t!]
\caption{Over-Approximation Routine.}\label{algorithm_main}
\begin{algorithmic}[1]

\REQUIRE An agent network $(\cals,\calK,\calF)$ and uncertainty set $\calU_{\calK}^\delta$, a finite time horizon $T > 0$, some (small) positive function $\eps^{(0)}$ and a numerical threshold $\eta \in (0;1)$.

\ENSURE Formal bound of $\mathcal{E}$ from Definition~\ref{def_max_dev}.

\STATE \textbf{compute} the transition rates $(r_{B,C})_{B,C}$ of the
\STATE \qquad \qquad decoupled CTMC from Definition~\ref{def_atomic_ode_with_uu}

\IF{$(r_{B,C})_{B,C}$ violates \textbf{(A1)}}
    \STATE \textbf{compute} an envelope $(\tr_{B,C})_{B,C}$ of $(r_{B,C})_{B,C}$ \label{alg_line_env}
%    \STATE \qquad \qquad using Algorithm~\ref{algorithm_envelope}
\ELSE
    \STATE \textbf{set} $(\tr_{B,C})_{B,C}$ to $(r_{B,C})_{B,C}$ and $\calI$ to $\emptyset$
\ENDIF

\STATE \textbf{compute} the coarsening $(\hr_{B,C})_{B,C}$ of $(\tr_{B,C})_{B,C}$ \label{alg_line_coars}
%\STATE \qquad \qquad as given in Definition~\ref{def_coars}

\STATE \textbf{set} $\epso$ to zero $\eps^{(0)}$
\WHILE{\TRUE}
    \STATE \textbf{compute} $\Psi(\epso)$ from Theorem~\ref{thm_ctmc_estimator} using Theorem~\ref{thm_suff_cond_for_suff_pont} \label{alg_line_psi}
    \STATE \textbf{set} $\epsn$ to $\Psi(\epso)$
    \IF{\textbf{not }($\epsn < V^0$)}
        \RETURN $\infty$
    \ELSIF{\big($\eta \geq \max_{A \in \cals, t \in [0;T]} |\epsn_A(t) - \epso_A(t)|$\big)}
        \RETURN $\epsn$
    \ENDIF
    \STATE \textbf{set} $\epso$ to $\epsn$
\ENDWHILE

\end{algorithmic}
\end{algorithm}

\begin{proof}
%Since this implies that any solution $\pi^{\fru}$ of~(\ref{eq_atomic_ode_with_uu}) is a solution of~(\ref{eq_tilde_pi}), we conclude that $\Phi(\eps) \leq \Psi(\eps)$ for all $0 \leq \eps < V^0$.
To see that $\Psi(\eps) \leq \eps$ implies $\calE \leq \eps$, let $\eps$, $\delta$, $\hk$, $u_\calK$, $u_\cals$, $\fru$, $\tha$ and $A$ be as in the proof from Section~\ref{sec_thm_main_bound}. Then, it holds that $\eps_A(\tha) = |\pi_A^{\fru}(\tha) - \pi_A^0(\tha)| \leq |\tilde{\pi}_A^\ast(\tha) - \tilde{\pi}_A^0(\tha)| = (\Psi_A(\eps))(\tha) \leq \eps_A(\tha)$, where $\tilde{\pi}_A^\ast(\tha)$ solves
\begin{multline*}%\label{eq_opt_proof_2}
\text{compute the maximal value of $\tilde{\pi}^{(u'_{\calK},u'_\cals,u'_\calI)}_A(\tha)$ } \\
\text{ such that } \dot{\tilde{\pi}}(t) = \tilde{h}\big(t,\tilde{\pi}(t),(u'_{\calK}(t),u'_\cals(t),u'_\calI(t))\big),  \\
\text{ (\ref{eq_init_pi}) and $(u'_{\calK},u'_\cals,u'_\calI) \in \calU_{\calK}^\delta \times \calU_\cals^{\eps} \times \calU_\calI^{b}$}
\end{multline*}
and the inequality $|\pi_A^{\fru}(\tha) - \pi_A^0(\tha)| \leq |\tilde{\pi}_A^\ast(\tha) - \tilde{\pi}_A^0(\tha)|$ holds true because the definition of the envelope and the coarsening of an envelope ensure that for any function $(u_{\calK},u_\cals) \in \calU_{\calK}^{\delta} \times \calU_\cals^\eps$ there exists some function $u_\calI \in \calU_\calI^{b}$ such that
\begin{multline*}
r_{B,C}\big(V^0(t) + u_\cals(t), \hk(t) + u_\calK(t)\big) \\
= k^{B \to C}(t) + \sum_{i \in \calK \cup \cals \cup \calI} k^{B \to C}_i(t) u_i(t) \\
= \tr_{B,C}(t,u_\calK(t),u_\cals(t),u_\calI(t))
\end{multline*}
for all $0 \leq t \leq T$. This implies $\tilde{\pi}_A^\ast(\tha) = \pi_A^{\fru}(\tha)$. Moreover, it can be observed that Theorem~\ref{thm_suff_cond_for_suff_pont} and Theorem~\ref{prop_lip} hold for any CTMC which transition rates satisfy~\textbf{(A1)} and~\textbf{(A2)}. Hence, they can be used to compute $\tilde{\pi}_A^\ast(\tha)$ (just replace the index set $\calK \cup \cals$ with $\calK \cup \cals \cup I$). This and the definition of the envelope and the coarsening of an envelope ensure the existence of some $u_\calI \in \calU_\calI^{b}$ such that $\frut = (u_\calK,u_\cals,u_\calI)$ is optimal, i.e., $\tilde{\pi}_A^\ast(\tha) = \tilde{\pi}_A^{\frut}(\tha)$. In the case $b_i(\cdot)$ is piecewise constant for all $i \in \calI$, the argumentation from Section~\ref{sec_thm_main_bound} leads to the desired contradiction. With this, the density argument from Section~\ref{sec_thm_main_bound} implies that $\calE \leq \eps$. Since  $\Psi$ is monotonic increasing, the proof is complete.
\end{proof}

It is interesting to note that Theorem~\ref{thm_main_bound} allows one to derive bounds on $\mathcal{E}$ from Definition~\ref{def_max_dev} using any estimation technique that applies to~(\ref{eq_opt}). Put different, Theorem~\ref{thm_suff_cond_for_suff_pont} and~\ref{prop_lip} can be replaced, in principle, by any over-approximation technique applicable to time-varying linear systems with uncertain additive and multiplicative uncertainties. Note, however, that the bounds obtained by Theorem~\ref{thm_suff_cond_for_suff_pont} and~\ref{prop_lip} cannot be improved if the decoupled CTMC satisfies $\textbf{(A1)}$-$\textbf{(A2)}$ and can be expected to be tight even if $\textbf{(A1)}$-$\textbf{(A2)}$ are violated because Theorem~\ref{thm_ctmc_estimator} relies on optimal control theory.

%Moreover, while there is a rich body of literature covering over-approximation techniques of linear ODE systems (see~\cite{Kurzhanski2000,Girard2006,B-GirLeG08a} and references therein), to the best of our knowledge, only~\cite{Althoff2011a} can be used to estimate~(\ref{eq_opt}). The bounds of the more general approach~\cite{Althoff2011a} however can be expected to be less tight on the domain of ICTMDPs than those of Theorem~\ref{thm_suff_cond_for_suff_pont}. To see this, assume that the transition matrix of the CTMC has the time-varying transition rate function $1.00 + 0.50 \sin(t) + u(t)$, where $u$ denotes an uncertainty with $|u(\cdot)| \leq 0.05$. Then, the underlying interval matrix from~\cite{Althoff2011a} will have the conservative interval entry $[0.45;1.55]$ and, consequently, induce less tight bounds than Theorem~\ref{thm_suff_cond_for_suff_pont}.

\begin{algorithm}[t!]
\caption{Envelope computation for agent networks that have as reaction rate functions $\calF$ multivariate polynomials.}\label{algorithm_envelope}
\begin{algorithmic}

\REQUIRE Transition rates $(r_{B,C})_{B,C}$ given in terms of polynomials with variables $\{V_A^0(t) \mid A \in \cals\} \cup \{u_i \mid i \in \cals \cup \calK\}$
\ENSURE Envelope $(\tr_{B,C})_{B,C}$ of $(r_{B,C})_{B,C}$, index set $\calI$ of new uncertainties

\STATE \textbf{set} $\calI$ to $\emptyset$
\FORALL{$B,C \in \cals$}
    \FOR{\textbf{each} nonlinear uncertainty $\prod_{i \in \calK \cup \cals} u_i^{e_i}$ in $r_{B,C}$}
        \STATE \textbf{add} $e$ to $\calI$, where $e \in \mathbb{N}_0^{\calK \cup \cals}$ denotes the exponent
        \STATE \qquad of the current monomial
        \STATE \textbf{replace} $\prod_{i \in \calK \cup \cals} u_i^{e_i}$ by the new uncertainty $u_e$
        \STATE \textbf{set} $\frb_e(\cdot)$ to $\prod_{i \in \calK \cup \cals} \frb_i^{e_i}(\cdot)$
    \ENDFOR
\ENDFOR

\RETURN $(r_{B,C})_{B,C}$ and $\calI$

\end{algorithmic}
\end{algorithm}

\subsection{Algorithm}\label{sec_impl}

The previous sections gives rise to Algorithm~\ref{algorithm_main} which summarizes all steps of our approximation technique. Apart from line~\ref{alg_line_env} that has to compute an envelope of the transition rates of the decoupled CTMC from Definition~\ref{def_atomic_ode_with_uu} (see Definition~\ref{def_envelope}), all steps of Algorithm~\ref{algorithm_main} can be automatized. Indeed, the computation of the coarsening in line~\ref{alg_line_coars} is the variable substitution introduced in Definition~\ref{def_coars}, while $\Psi(\epso)$ in line~\ref{alg_line_psi} can be obtained by applying, for any $A \in \cals$ and $0 \leq \tha \leq T$, Theorem~\ref{thm_suff_cond_for_suff_pont} in order to compute the maximal and minimal value of $\tilde{\pi}^{\frut}_A(\tha)$ across all $\frut \in \calU^\delta_\calK \times \calU^\eps_\cals \times \calU^b_\calI$.

\emph{Computation of an envelope.} In the case the reaction rate functions of the agent network are multivariate polynomials as in the case of our running example discussed in Example~\ref{ex_sir}-\ref{exex_multi_sir}, the envelope from line~\ref{alg_line_env} can be efficiently computed by Algorithm~\ref{algorithm_envelope}. This makes our approach particularly suited to models from the field of biochemistry. Note that Algorithm~\ref{algorithm_envelope} replaces any product of uncertainties by a new uncertainty and bounds it by the maximal value of the replaced product similarly to the discussion following Example~\ref{ex_sir_atomic_uu_2}.

\emph{Computation of $\Psi$.} We conclude the section by discussing a rigorous and a heuristic implementation of line~\ref{alg_line_psi}. In the case of the former, one has to combine Theorem~\ref{thm_suff_cond_for_suff_pont},~\ref{prop_lip} and \ref{thm_numeric} with a verified numerical ODE solver as~\cite{BM07} that provides apart from a numerical solution also an estimation of the underlying numerical error~\cite{Gear:1971}. Additionally to the numerical error, one has to account for the discretization error arising in the computation of $\Psi(\eps)$, where the idea is to evaluate $(\Psi(\eps))(\cdot)$ \emph{only at grid points} $\tha_l$ from $\mathcal{T}(\dt) = \{0,\dt, 2\dt, \ldots, T\}$ by computing the maximal and minimal value of $\tilde{\pi}^{\frut}_A(\tha_l)$ from~(\ref{eq_tilde_pi}) for all $A \in \cals$ and $\tha_l \in \mathcal{T}(\dt)$.

The following result can be proven.

\begin{theorem}\label{thm_num_odes}
For any positive function $\eps < V^0$ and $\xi \in (0;1)$, the function $t \mapsto (\Psi(\eps))(t)$ can be approximated with precision $\xi$ by solving $4 |\cals| \Lambda T \xi^{-1}$ ODE systems of size $|\cals|$, where
\begin{multline}\label{eq_lambda}
\Lambda \geq \max\big\{ \norm{\tilde{h}(t,\tilde{\pi},(u_{\calK},u_\cals,u_\calI))} \mid 0 \leq t \leq T, \\ \norm{\tilde{\pi}} \leq \normo{V(0)} \text{ and }
|u_i| \leq \sup_{0 \leq t \leq T} \frb_i(t) \big\}
\end{multline}
and $\tilde{h}$ denotes the Kolmogorov equations underlying the transition rates $(\hr_{B,C})_{B,C}$ from Algorithm~\ref{algorithm_main}.
\end{theorem}

\begin{table*}[tp]
\centering
    \begin{tabular}{rrrrrrrrr}
    \toprule
    \multicolumn{1}{c}{\empty} & \multicolumn{1}{c}{\empty} & \multicolumn{1}{c}{\empty} & \multicolumn{2}{c}{CORA} & \multicolumn{2}{c}{Algorithm~\ref{algorithm_main} for $\mathcal{T}(0.04)$} & \multicolumn{2}{c}{Algorithm~\ref{algorithm_main} for $\mathcal{T}(0.03)$} \\
    \cmidrule(l){4-5} \cmidrule(l){6-7} \cmidrule(l){8-9}
    $D$ & $|\cals|$ & $|\calK|$ & Run time & $\stackrel{\text{\normalsize Bound on}}{\sup_t \norm{\mathcal{E}(t)}}$ & Run time & $\stackrel{\text{\normalsize Bound on}}{\sup_t \norm{\mathcal{E}(t)}}$ & Run time & $\stackrel{\text{\normalsize Bound on}}{\sup_t \norm{\mathcal{E}(t)}}$ \\
    \midrule
%    1 &  3 & 3  & 3s  & 0.060 & 4m  & 0.159 & 3m  & 0.163 \\
%    2 &  6 & 8  & 7s  & 0.173 & 4m  & 0.124 & 6m  & 0.130 \\
%    3 &  9 & 15 & --- & --- & 7m  & 0.110 & 10m & 0.116 \\
%%    3 &  9 & 15 & 13s & 0.322 & 7m  & 0.110 & 10m & 0.116 \\
%    4 & 12 & 24 & ---    & ---   & 12m & 0.101 & 15m & 0.109 \\
%    5 & 15 & 35 & ---    & ---   & 17m & 0.097 & 23m & 0.103 \\
%    6 & 18 & 48 & ---    & ---   & 22m & 0.095 & 28m & 0.101 \\

    1 &  3 & 3  & 0m  & 0.187 & 2m  & 0.158 & 3m  & 0.163 \\
    2 &  6 & 8  & 1m  & 0.238 & 4m  & 0.124 & 6m  & 0.130 \\
    3 &  9 & 15 & 1m  & 0.296 & 8m & 0.109 & 11m & 0.116 \\
    4 & 12 & 24 & 2m  & 0.377 & 13m & 0.100 & 17m & 0.109 \\
    5 & 15 & 35 & 4m  & 0.494 & 18m & 0.096 & 24m & 0.103 \\
    6 & 18 & 48 & 5m & 0.694   & 19m & 0.091 & 29m & 0.101 \\
    7 & 21 & 63 & --- & ---   & 30m & 0.091 & 43m & 0.094 \\
    8 & 24 & 80 & --- & ---   & 40m & 0.084 & 53m & 0.096 \\
    9 & 27 & 99 & --- & ---   & 50m & 0.087 & 65m & 0.096 \\
    10 & 30 & 120 & --- & ---   & 61m & 0.091 & 78m & 0.095 \\
    \bottomrule
    \end{tabular}
\caption{Results obtained by applying CORA and the heuristic implementation of Algorithm~\ref{algorithm_main} to the SIRS model~(\ref{eq_sir_case}). While CORA terminated for $D \leq 6$, no estimations could be obtained in the case of $D \geq 7$ due to out-of-memory errors. The heuristic implementation of Algorithm~\ref{algorithm_main}, instead, is slower than CORA but scales to larger systems and provides tight bounds.}\label{tab_case}
\end{table*}

\begin{proof}
We prove that we need to solve $2 |\cals| \Lambda T \xi^{-1}$ instances of~(\ref{eq_pont_cost}) and~(\ref{eq_pont_pi}), respectively. To this end, we first note that $\tilde{\pi}^{\frut}$ from~(\ref{eq_tilde_pi}) is absolutely continuous and has derivative $\tilde{h}(\cdot,\tilde{\pi}^\frut,\frut)$ almost everywhere. This yields
\[
\tilde{\pi}^{\frut}(t_2) - \tilde{\pi}^{\frut}(t_1) = \int_{t_1}^{t_2} \tilde{h}(s,\tilde{\pi}^\frut(s),\frut(s)) ds
\]
for any $0 \leq t_1 \leq t_2 \leq T$. Since $\norm{\tilde{\pi}^{\frut}(t)} \leq \normo{V(0)}$ for all $0 \leq t \leq T$, we infer that $\norm{\tilde{\pi}^{\frut}(t_2) - \tilde{\pi}^{\frut}(t_1)} \leq \Lambda |t_2 - t_1|$. This implies that we miss the actual value of $(\Psi_\cdot(\eps))(\cdot)$ by at most $\Lambda \dt$ if we compute the maximal and minimal value of $\tilde{\pi}^{\frut}_A(\tha_l)$ for all $A \in \cals$ and $\tha_l \in \mathcal{T}(\dt) = \{0,\dt, 2\dt, \ldots, T\}$. With this, we note that $\Lambda \dt \leq \xi$ implies $\dt \leq \xi / \Lambda$ which, in turn, induces $T / \dt = \Lambda T \xi^{-1}$ grid points. Since we need to compute the minimum and maximum value of $\tilde{\pi}^\fru_A(\tha_l)$ for all $A \in \cals$ and $\tha_l \in \mathcal{T}(\dt)$, this yields the claim.
\end{proof}

Apart from the rigorous implementation, our approach allows for a heuristic implementation. Here, instead of using a verified numerical ODE solver, one invokes a standard ODE solver in which the numerical error is minimized heuristically by varying the integration step size~\cite{Gear:1971}. Similarly, one accounts heuristically for the discretization underlying $\mathcal{T}(\dt)$ by gradually refining an initially coarse discretization of $[0;T]$ until the approximations of $\Psi$ are reasonably close.

We wish to point out that both implementations naturally apply to parallelization because each single ODE system can be solved independently from the others.

\section{Numerical Evaluation}\label{sec_case_studies}

In this section we study the potential of our technique by applying it on the multi-class SIRS model from Section~\ref{sec_suboptimal}. To this end, we implemented an experimental prototype of the heuristic version from Section~\ref{sec_impl} in Matlab by relying on the (non-verified) numerical ODE solver provided by the Matlab command \texttt{ode45s}. The heuristic implementation was compared with the state-of-the-art reachability analysis tool CORA~\cite{Althoff2015a} that covers nonlinear ODE systems with multiplicative uncertainty functions. %(We did not compare to Flow$^*$~\cite{DBLP:conf/cav/ChenAS13} which supports time-varying uncertainties as well because it has been compared to CORA in~\cite{Althoff2013a}.)

All experiments were performed on a 3.2 GHz Intel Core i5 machine with 8 GB of RAM. The Matlab solver \texttt{ode45s} was invoked with its default settings. For CORA, instead, the time step was set to $0.004$, while the expert settings were chosen as in the nonlinear tank example from the CORA manual~\cite{CoraManual}. The main findings are as follows.
\begin{itemize}
    \item The bounds obtained by the heuristic implementation are tight;
    \item CORA is faster than the heuristic implementation in the case of smaller systems;
    \item The heuristic implementation scales to models that cannot be covered by CORA.
\end{itemize}

The above confirms the discussion from Section~\ref{sec_impl} concerning the complexity of our approach. Indeed, for smaller ODE systems, our approach is inferior to CORA because of the discretization of the time interval $[0;T]$. However, for abstraction approaches such as CORA it is computationally prohibitive to obtain tight over-approximations for larger nonlinear systems in general. Instead, our approach requires to solve $4|\cals| \Lambda T \xi^{-1}$ ODE systems of size $|\cals|$, see Theorem~\ref{thm_num_odes}. Moreover, at least as far as the heuristic implementation is considered, we were able to obtain tight bounds. In summary, we argue that our technique has the potential to complement state-of-the-art over-approximation approaches.

\emph{Multi-class SIRS Model.} The global dynamics underlying~(\ref{ex_multi_sir}) is given by
\begin{align}\label{eq_sir_case}
\dot{V}^{u_\calK}_{S_\nu} & = - \sum_\mu (\hk_{\alpha_{\nu,\mu}} + u_{\alpha_{\nu,\mu}}) V^{u_\calK}_{S_\nu} V^{u_\calK}_{I_\mu} + (\hk_{\gamma_\nu} + u_{\gamma_\nu}) V^{u_\calK}_{R_\nu} \nonumber \\
\dot{V}^{u_\calK}_{I_\nu} & = - (\hk_{\beta_\nu} + u_{\beta_\nu}) V^{u_\calK}_{I_\nu} + \sum_\mu (\hk_{\alpha_{\nu,\mu}} + u_{\alpha_{\nu,\mu}}) V^{u_\calK}_{S_\nu} V^{u_\calK}_{I_\mu} \nonumber \\
\dot{V}^{u_\calK}_{R_\nu} & = - (\hk_{\gamma_\nu} + u_{\gamma_\nu}) V^{u_\calK}_{R_\nu} + (\hk_{\beta_\nu} + u_{\beta_\nu}) V^{u_\calK}_{I_\nu}
\end{align}
where $u_{\alpha_{\nu,\mu}}$, $u_{\beta_\nu}$, $u_{\gamma_\nu}$ are \emph{time-varying uncertain} functions and $1 \leq \nu, \mu \leq D$ for some $D \geq 1$. The system has $3D$ ODE variables and $D^2 + 2D$ uncertainties. As discussed in Section~\ref{sec_suboptimal}, the transition rates (\ref{ex_sir_multi_rates}) define a function $\Psi$ as in Theorem~\ref{thm_ctmc_estimator} that can be used to estimate the function $\calE$ underlying~(\ref{eq_sir_case}). % from Definition~\ref{def_max_dev} that underlies~(\ref{eq_sir_case}).

In our experiments, we randomly chose $\hk_{\alpha_{\nu,\mu}} \equiv 1.00$, $\hk_{\beta_\nu} \equiv 2.00$, $\hk_{\gamma_\nu} \equiv 3.00$ and $V^{u_\calK}_{S_\nu}(0) = 4.00 + 0.10 (\nu - 1)$, $V^{u_\calK}_{I_\nu}(0) = V^{u_\calK}_{R_\nu}(0) = 1.00$ for all $1 \leq \nu,\mu \leq D$. The time horizon was set to $T = 3.00$, while all parameters were subject to uncertainties with modulus not higher than $\zeta = 0.03$, i.e., $\frb_\theta(t) = \zeta$ for all $\theta \in \calK$ and $0 \leq t \leq T$.

Table~\ref{tab_case} and Figure~\ref{fig_reachset_bounds} summarize our findings. With increasing $D$, the tightness of the bounds provided by CORA decreases while the corresponding running times increase. In principle, the tightness can be improved by using stricter parameters (e.g., by decreasing the step size). This, however, increases the time and space requirements. Likewise, the over-approximation of larger models requires more resources in general. On our machine, for instance, $D \geq 7$ or time steps below $0.004$ led to out-of-memory errors. The heuristic implementation of Algorithm~\ref{algorithm_main}, instead, scales to larger instances of the running example and provides tight bounds. Indeed, since the $\Lambda$ from~(\ref{eq_lambda}) can be chosen as $6 D$ in the case of the ODE system~(\ref{ex_multi_sir}), Theorem~\ref{thm_num_odes} implies that one has to solve $4|\cals| \Lambda T \xi^{-1} = 216 D^2 \xi^{-1}$ ODE systems of size $3D$ in order to guarantee that a numerical approximation of $\Psi(\eps)$ from Theorem~\ref{thm_ctmc_estimator} misses the actual value of $\Psi(\eps)$ by at most $\xi > 0$.
We approximated the values of $\Psi$ from Algorithm~\ref{algorithm_main} using discretizations $\mathcal{T}(0.04)$ and $\mathcal{T}(0.03)$, where
\[
\mathcal{T}(\dt) = (\{ l \dt \mid l \geq 0 \} \cup \{3\}) \cap [0;3]
\]
The run times account for the computation of the sequence $(\eps^{(k)})_k$ from Algorithm~\ref{algorithm_main} with $\eta = 10^{-4}$. In agreement with Theorem~\ref{thm_num_odes}, the running times exhibit a polynomial growth. Moreover, discretizations $\mathcal{T}(0.04)$ and $\mathcal{T}(0.03)$ induce bounds that are reasonably close.

\begin{figure}[!t]
\centering
\hspace{-1.6cm}
\subfloat{%
\includegraphics[width=0.18\textwidth]{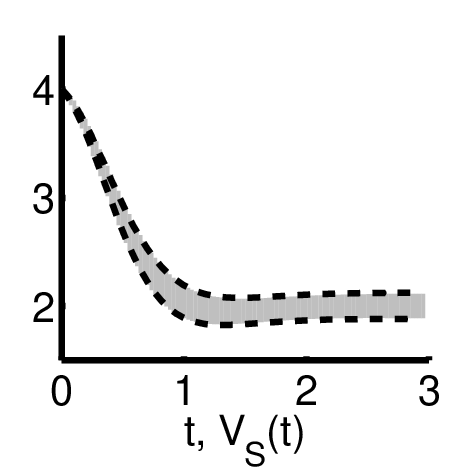}
}
\hspace{-0.4cm}
\subfloat{%
\includegraphics[width=0.18\textwidth]{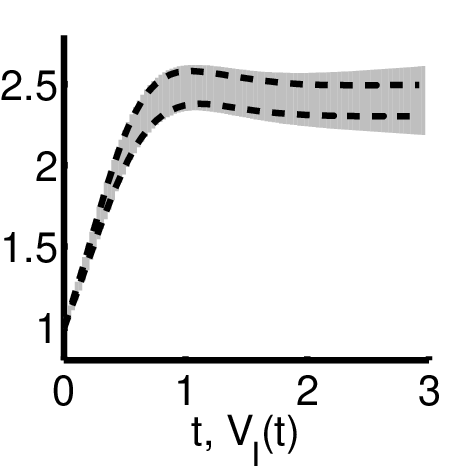}
}
\hspace{-0.4cm}
\subfloat{%
\includegraphics[width=0.18\textwidth]{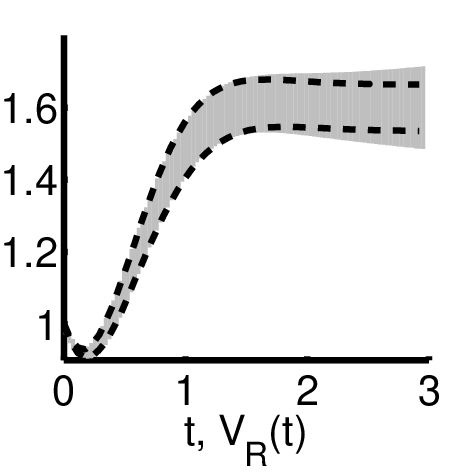}
}
\hspace{-1.3cm}
\caption{Reachable set estimation of~(\ref{eq_sir_case}) in case of $D = 1$ and $\zeta = 0.03$. The gray area visualizes the estimation of CORA, while the dotted lines depict the one underlying the heuristic version of Algorithm~\ref{algorithm_main} for $\mathcal{T}(0.04)$. It can be seen that both techniques provide tight bounds, even though those of CORA appear to become less strict as time increases.}\label{fig_reachset_bounds}
\end{figure}

%The ODE systems~(\ref{eq_pont_cost}) and~(\ref{eq_pont_pi}) were solved using the Matlab command \texttt{ode45s} which implements a Runge-Kutta method of order four with adaptive step size~\cite{Gear:1971}.

\section{Discussion}\label{sec_disc}

While Pontryagin's principle and its extensions to systems with uncertain parameters have been used in the context of reachability analysis~\cite{dsn16BortolussiGast,doi:10.1137/0325010}, to the best of our knowledge, the principle has not been applied in the context of formal over-approximation of a general class of nonlinear ODE systems. This is because the principle is in general only a necessary condition for optimality, while its strict versions~\cite{KAMIEN,4522600} require concavity or convexity which is rarely satisfied by nonlinear ODE models. Additionally, Pontryagin's principle induces in general a differential inclusion which can only be solved under additional assumptions~\cite{NumericDiffInc1,DBLP:conf/qest/Bortolussi11}. The present work addresses those problems by a) approximating the original nonlinear ODE system by a family of linear Kolmogorov equations~(\ref{eq_atomic_ode_with_uu}) with multiplicative and additive uncertainties and; by b) showing that each family member can be over-approximated tightly and efficiently using a modified version of the strict version of Pontryagin's principle~\cite{KAMIEN}.

The proposed approach is complementary to existing approximation techniques. Indeed, while it is less efficient than approaches that are based on monotonic systems and differential inequalities~\cite{Ramdani2008,RAMDANI2010263,Scott201393,DBLP:journals/tac/TschaikowskiT16}, it may provide tighter bounds because it relies on optimal control theory. Instead, for approaches based on abstraction~\cite{DBLP:conf/cav/ChenAS13,Althoff2015a} and the Hamilton-Jacobi equation~\cite{DBLP:journals/automatica/Lygeros04,mitchell2005time}, in general it becomes computationally prohibitive to obtain tight over-approximations for larger nonlinear systems~\cite{Althoff2013a,DBLP:conf/cav/Duggirala016}. Another point worth stressing is that many approaches applicable to nonlinear ODE models assume time-invariant uncertain parameters and uncertain initial conditions, while the present technique focusses on nonlinear ODE systems with time-varying uncertain parameters and fixed initial conditions. Since the proposed approximation technique relies on the availability of a concrete nominal solution, a direct extension to sets of initial conditions seems not to be possible. This notwithstanding we wish stress that it is particulary suited to systems biology where initial concentrations can be measured while reaction rates are often difficult to obtain and may vary with time.

%While abstraction techniques can cover many practical models, in general it is computationally prohibitive to obtain tight over-approximations for larger nonlinear systems~\cite{Althoff2013a,DBLP:conf/cav/Duggirala016}. We shall refer to this phenomenon as the curse of dimensionality.

\section{Conclusion}\label{sec_conclusion}

In this work we presented an over-approximation technique for nonlinear ODE systems with time-varying uncertain parameters. Our approach provides verifiable bounds in terms of a family of linear Kolmogorov equations with uncertain additive and multiplicative time-varying parameters. To ensure efficient computation and tight estimations, we have established, to the best of our knowledge, a novel efficiently computable solution technique for a class of inhomogeneous continuous time Markov decision processes.

The presented over-approximation technique is efficient and can be expected to provide tight bounds because it relies on optimal control theory and allows for an algorithmic treatment in the case where the ODE system is given by multivariate polynomials. This makes it particularly suited to models from (bio)chemistry.

By comparing our approach with a state-of-the-art over-approximation technique in the context of the multi-class SIRS model from epidemiology~\cite{dsn16BortolussiGast}, we have provided numerical evidence for the potential of our approach. The most pressing line of future work is the development of a tool which provides a rigorous implementation of the technique.

\section*{Acknowledgement} The author thanks Mirco Tribastone for helpful discussions. Parts of the work have been conducted when the author was with IMT Lucca. The author is supported by a Lise Meitner Fellowship that is funded by the Austrian Science Fund (FWF) under grant number M 2393-N32 (COCO).

\bibliographystyle{IEEEtran}

\bibliography{root}
\vspace{-2.5cm}

\begin{IEEEbiography}[{\includegraphics[width=1in,height=1.25in,clip,keepaspectratio]{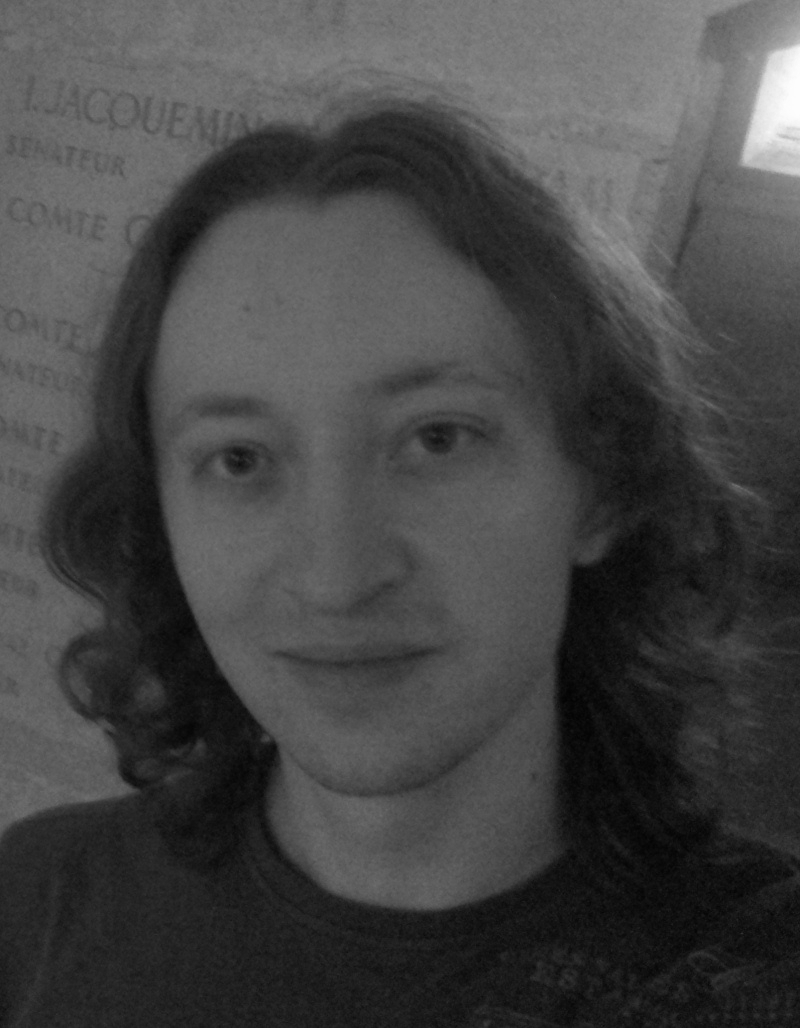}}]{Max Tschaikowski}
is a Lise Meitner Fellow at the Technische Universit\"{a}t Wien, Austria. Prior to it, he was a non-tenure-track Assistant Professor at IMT Lucca, Italy, a Research Fellow at the University of Southampton, UK, and a Research Assistant at the LMU in Munich, Germany. He was awarded a Diplom in mathematics and a Ph.D. in computer science by the LMU in 2010 and 2014, respectively. His research focusses on the construction, reduction and verification of quantitative models.
\end{IEEEbiography}

\end{document}